\documentclass[11pt]{article}

\usepackage[paper=a4paper,left=29mm,right=29mm,top=28mm, bottom=28mm]{geometry}

\usepackage{hyperref}

\usepackage{academicons}

\usepackage{xcolor}

\usepackage[version=4]{mhchem}

\usepackage[ruled,vlined]{algorithm2e}

\usepackage{tikz,tikz-qtree}
\usetikzlibrary{positioning,arrows,arrows,automata,shapes}
\tikzset{
  invisible/.style={opacity=0},
  visible on/.style={alt={#1{}{invisible}}},
  alt/.code args={<#1>#2#3}{%
    \alt<#1>{\pgfkeysalso{#2}}{\pgfkeysalso{#3}} % \pgfkeysalso doesn't change the path
  },
}

\usepackage{authblk}
\usepackage{tikz}
\usepackage{xfrac}
\usepackage{tabularx, multirow}
\usepackage{graphicx, tikz}
\usetikzlibrary{arrows.meta}
\usetikzlibrary{arrows,calc}
\usetikzlibrary{shapes}

\usepackage{latexsym,amssymb,amsfonts,graphicx,makeidx,paralist}
\usepackage{theorem,rotating,ifthen,amsmath,subfigure,epsfig,nicefrac}
\usepackage{booktabs,framed}

\usepackage{color,framed, float,wasysym,cancel}

\colorlet{shadecolor}{gray!12}

\newcommand{\gr} {\mbox{Digraph}}

\newcommand{\g} {\mbox{digraph}}

\newenvironment{desctight}
  {\begin{list}{}{\setlength\labelwidth{0pt}
        \setlength{\itemsep}{0.5pt}
        \setlength{\parsep}{0pt}
        \setlength\itemindent{-\leftmargin}
        }}
    {\end{list}}

\newtheorem{theorem}{Theorem}[section]
\newtheorem{example}[theorem]{Example}
\newtheorem{corollary}[theorem]{Corollary}
\newtheorem{lemma}[theorem]{Lemma}
\newtheorem{observation}[theorem]{Observation}

\newtheorem{definition}[theorem]{Definition}

\newtheorem{remark}[theorem]{Remark}
\newtheorem{proposition}[theorem]{Proposition}

\newenvironment{proof}{\noindent{\bf Proof~}}{\null\hfill $\Box$\par\medskip}

\newcommand{\bigo}{\text{$\mathcal O$}}

\newcommand{\ideg}{\text{indegree}}
\newcommand{\odeg}{\text{outdegree}}
\newcommand{\un} {{\it un}}

\newcommand{\CN} {\text{CN}}

\newcommand{\OCN} {\text{OCN}}

\newcommand{\OCI} {\text{OCI}}

\begin{document}

\title{Efficient computation of oriented vertex and arc colorings of special digraphs\thanks{A short version of this paper will appear in the Proceedings of the {\em International Conference on Operations Research} (OR 2021) \cite{LGK21}.}}

\author[1]{Frank Gurski}
\author[1]{Dominique Komander}
\author[1]{Marvin Lindemann}

\affil[1]{\small University of  D\"usseldorf,
Institute of Computer Science, Algorithmics for Hard Problems Group,
40225 D\"usseldorf, Germany}

\maketitle

%%%%%%%%%%%%%%%%%%%%%%%%%%%%%%%%%%%%%%%%%%%%%%%%%%%%%%%%%%%%%%%%%%%%%%%%%%%
%%%%%%%%%%%%%%%%%%%%%%%%%%%%%%%%%%%%%%%%%%%%%%%%%%%%%%%%%%%%%%%%%%%%%%%%%%%
%%%%%%%%%%%%%%%%%%%%%%%%%%%%%%%%%%%%%%%%%%%%%%%%%%%%%%%%%%%%%%%%%%%%%%%%%%%

\begin{abstract}
In this paper we study the oriented vertex and arc coloring problem on edge
series-parallel digraphs (esp-digraphs) which are related to the well known series-parallel graphs.
Series-parallel graphs
are graphs with two distinguished vertices called terminals,
formed recursively by parallel and series composition. These graphs have applications
in  modeling series and parallel electric circuits and also
play an important role in theoretical computer science.
The oriented class of series-parallel digraphs is recursively defined from pairs of vertices
connected by a single arc and applying the parallel and series
composition, which leads to
specific orientations of undirected series-parallel graphs. 
Further we consider the line digraphs of edge
series-parallel digraphs, which are known as minimal series-parallel digraphs (msp-digraphs).

We show tight upper bounds for the oriented chromatic number
and the oriented chromatic index of edge series-parallel digraphs and 
minimal series-parallel digraphs.
Furthermore, we introduce  first linear time solutions for 
computing the  oriented chromatic number of edge series-parallel digraphs and the  oriented chromatic index of minimal series-parallel digraphs.

\end{abstract}

%%%%%%%%%%%%%%%%%%%%%%%%%%%%%%%%%%%%%%%%%%%%%%%%%%%%%%%%%%%%%%%%%%%%%%%%%%
\section{Introduction}
%%%%%%%%%%%%%%%%%%%%%%%%%%%%%%%%%%%%%%%%%%%%%%%%%%%%%%%%%%%%%%%%%%%%%%%%%%

A homomorphism from an oriented graph $G=(V_G,E_G)$ to an oriented graph $H=(V_H,E_H)$
is an arc preserving
mapping $h$ from $V_G$ to $V_H$, i.e. if $(u,v)\in E_G$ then  $(h(u),h(v))\in E_H$.

An {\em oriented $r$-vertex-coloring} of an oriented graph $G$ corresponds to an oriented graph $H$
on $r$ vertices, such that there exists a homomorphism from $G$ to $H$.
The {\em oriented chromatic number} of $G$, denoted by $\chi_o(G)$, is the minimum number of vertices in an
oriented graph $H$ such that there is a homomorphism from $G$ to $H$.

In the Oriented Chromatic Number problem ($\OCN$ for short) there is given an
oriented graph $G$ and an integer $r$ and we have to
decide whether there is an  oriented $r$-vertex-coloring for $G$.
If $r$ is constant, i.e.\ not part of the input, the corresponding problem
is denoted by $\OCN_{r}$. Even $\OCN_{4}$ is NP-complete~\cite{CD06}.

Moreover, an {\em oriented $r$-arc-coloring} of an oriented graph $G$ relates
to an oriented graph $H$
on $r$ vertices, such that there is a homomorphism from line digraph $LD(G)$ to $H$.
The {\em oriented chromatic index} of $G$, denoted by $\chi'_o(G)$, is the minimum number of vertices in an
oriented graph $H$ such that there is a homomorphism from line digraph $LD(G)$ to $H$.

In the Oriented Chromatic Index problem ($\OCI$ for short) there is given an
oriented graph $G$ and an integer $r$ and we have to
decide whether there is an  oriented $r$-arc-coloring for $G$.
If $r$ is constant, i.e.\ not part of the input, the corresponding problem
is denoted by $\OCI_{r}$. Even $\OCI_{4}$ is NP-complete \cite{OPS08}.

The hardness of $\OCN_{4}$ and  $\OCI_{4}$ motivates to consider the
oriented chromatic number and the oriented chromatic index of special graph classes.
This was frequently done
%The definition of oriented vertex-coloring and oriented arc-coloring was often
%used
for undirected graphs, where the
maximum value $\chi_o(G')$ or  $\chi'_o(G')$ of all possible orientations $G'$ of a graph
$G$ is considered, see \cite{DS14,Mar13,Mar15,OP14,Sop97} and \cite{OPS08,PS06}.
In this sense it has been shown in \cite{Sop97} that every
series-parallel graph has oriented chromatic number at most $7$ and that this bound is tight.
%
%Further, due \cite{PS06} for  every
%series-parallel graph we also have a tight bound since the oriented chromatic index is at most $7$.
%
Since the oriented chromatic index is always less or equal the oriented chromatic 
number (Observation \ref{obs-ind-chr}) every series-parallel graph has 
oriented chromatic index at most $7$ and by \cite{PS06} this bound is also tight.

%% esp
In this paper we consider the oriented chromatic number and the oriented chromatic 
index of esp-digraphs (short for edge series-parallel digraphs) which can recursively be defined from
the single edge graph by applying the  parallel composition and series composition.
and lead to specific orientation of  series-parallel graphs.
%
%%% msp
Further, we consider the oriented chromatic number and the oriented chromatic index of line digraphs of edge
series-parallel digraphs, which are known as msp-digraphs
(short for minimal series-parallel digraphs) and can recursively
be defined from
the single vertex graph by applying the  parallel composition and series composition.
The classes of msp-digraphs and esp-digraphs are incomparable in terms of  
set inclusion (Remark \ref{rem-inclu}).

%%%% results

As every  esp-digraph is a specific orientation of a series-parallel graph
the  mentioned bounds for  series-parallel graphs
lead to (not necessarily tight) upper bounds for the oriented chromatic number
and the  oriented chromatic index of esp-digraphs.

In this paper we re-prove the bound of $7$ for the oriented chromatic number
and the oriented chromatic index of esp-digraphs and we show
that these bounds are tight even for esp-digraphs. Since the oriented
chromatic index of esp-digraphs equals the oriented chromatic number of
their line digraphs, namely msp-digraphs, we obtain a tight upper bound
of $7$ for the oriented chromatic number of msp-digraphs. Further, (by Observation \ref{obs-ind-chr}) 
this leads to an upper bound of $7$ for the oriented chromatic index of msp-digraphs. 
We give an example that this bound is best possible.

We also consider solutions for computing the
oriented chromatic number and the oriented chromatic index of esp-digraphs and msp-digraphs.
In \cite{GKL20} we gave a first linear time solution for computing the  oriented chromatic number
of msp-digraphs. Using the line digraphs this leads to a linear time solution for computing the  oriented chromatic index of esp-digraphs. In this paper we introduce  linear time solutions for 
computing the  oriented chromatic number of esp-digraphs and the   oriented chromatic index of 
msp-digraphs.

%Furthermore, we give first linear time solutions for computing the  oriented chromatic number and the
%oriented chromatic index of esp-digraphs. 

%$\ldots$

In Tables \ref{tab1} and \ref{tab2}
we summarize these results.

\begin{table}
\begin{center}
\begin{tabular}{|c|l|l|ll|}
\hline
   $G$           &  $\chi_o(G)\leq 7$ &  sharpness &  \multicolumn{2}{c|}{recognition}   \\
\hline
esp-digraph & Theorem \ref{esp-bd7}   or Corollary \ref{cor-78}   &     Example \ref{ex-six}              &  $\bigo(n+m)$ & Theorem    \ref{esp-num-com} \\
\hline
msp-digraph &  \cite{GKL20}  or  Remark \ref{rem-msp-chi}              &  Example \ref{ex-sixa}   & $\bigo(n+m)$ & \cite{GKL20}\\
\hline
\end{tabular}
\end{center}

\caption{Oriented chromatic number $\chi_o$  of msp-digraphs and esp-digraphs}
\label{tab1}
\end{table}

\begin{table}
\begin{center}
\begin{tabular}{|c|l|l|ll|}
\hline
     $G$         & $\chi'_o(G)\leq 7$ &  sharpness &  \multicolumn{2}{c|}{recognition}   \\
     \hline
esp-digraph &  Remark \ref{rem-esp-chiprime}   or Corollary \ref{cor-77}                                 & Example \ref{ex-six2}  &   $\bigo(n+m)$  &  Corollary \ref{esp-index-com} \\ 
\hline     
msp-digraph &  Corollary \ref{cor-msp-ind}   &   Example \ref{ex-msp-6} &    $\bigo(n+m)$  &  Theorem \ref{msp-ori-a-c}   \\
\hline
\end{tabular}
\end{center}

\caption{Oriented chromatic index $\chi'_o$ of msp-digraphs and esp-digraphs}
\label{tab2}
\end{table}

%%%%%%%%%%%%%%%%%%%%%%%%%%%%%%%%%%%%%%%%%%%%%%%%%%%%%%%%%%%%%%%%%%%%%%%%%%%
\section{Preliminaries}\label{intro}
%%%%%%%%%%%%%%%%%%%%%%%%%%%%%%%%%%%%%%%%%%%%%%%%%%%%%%%%%%%%%%%%%%%%%%%%%%%

%%%%%%%%%%%%%%%%%%%%%%%%%%%%%%%%%%%%%%%%%%%%%%%%%%%%%%%%%%%%%%%%%%%%%%%%%%%
\subsection{Graphs and digraphs}
%%%%%%%%%%%%%%%%%%%%%%%%%%%%%%%%%%%%%%%%%%%%%%%%%%%%%%%%%%%%%%%%%%%%%%%%%%%

We refer to the notations of Bang-Jensen and Gutin \cite{BG09} for (di)graphs.
A {\em graph} is a pair  $G=(V,E)$ with
a finite set $V$ of {\em vertices}
and a finite set of {\em edges} $E \subseteq \{ \{u,v\} \mid u,v \in
V,~u \not= v\}$.
We call a pair  $G=(V,E)$ {\em directed graph} or {\em digraph}, such that $V$ is
a finite set of {\em vertices} and
$E\subseteq \{(u,v) \mid u,v \in V,~u \not= v\}$ is a finite set of ordered pairs of distinct
vertices called {\em arcs} or {\em directed edges}.
For a vertex $v\in V$, we define the sets $N^+(v)=\{u\in V \mid (v,u)\in E\}$ and
$N^-(v)=\{u\in V \mid (u,v)\in E\}$ as the {\em set of all successors}
and the {\em set of all  predecessors} of vertex $v$.
The  {\em outdegree} of $v$, $\odeg(v)$ for short, is the number
of successors of $v$ and analogously the  {\em indegree} of $v$, $\ideg(v)$ for short,
is the number of predecessors of $v$.

Digraph $G'=(V',E')$ is called a {\em subdigraph} of digraph $G=(V,E)$ if $V'\subseteq V$
and $E'\subseteq E$ holds. Moreover, $G'$ is an {\em induced subdigraph} of $G$,
denoted by $G'=G[V']$ if every arc of $E$ with both end vertices in $V'$ exists in $E'$.

For a digraph $G=(V,E)$ its {\em underlying undirected graph} is defined by disregarding the directions of the arcs, i.e., $\un(G)=(V,\{\{u,v\} \mid (u,v)\in E, u,v\in V\})$.

For an undirected graph $G=(V,E)$ replacing
every edge $\{u,v\}$ of $G$ by exactly one of the arcs $(u,v)$ and $(v,u)$ leads to an {\em orientation} of $G$.
Every digraph that can be obtained by an orientation of an undirected
graph $G$ is called an {\em oriented graph}, i.e.,
an oriented graph is a digraph without loops or opposite arcs.

A {\em tournament} is a digraph with exactly one arc between every two distinct vertices.
A digraph in which there are no directed cycles is called {\em directed acyclic graph (DAG for short)}.
The {\em girth} of digraph $G$ is defined by the length (number of arcs) of a shortest directed cycle in $G$.
For a DAG $G$ the girth is defined to be infinity.

The {\em line digraph} $LD(G)$ of digraph $G$ has a vertex for every
arc in $G$ and an arc from $u$ to $v$ if and only if $u=(x,y)$ and $v=(y,z)$ for vertices $x,y,z$ from $G$
\cite{HN60}. We call digraph $G$ the {\em root digraph} of $LD(G)$.

%%%%%%%%%%%%%%%%%%%%%%%%%%%%%%%%%%%%%%%%%%%%%%%%%%%%%%%%%%%%%%%%%%%%%%%%%%
\subsection{Undirected vertex-colorings}
%%%%%%%%%%%%%%%%%%%%%%%%%%%%%%%%%%%%%%%%%%%%%%%%%%%%%%%%%%%%%%%%%%%%%%%%%%

\begin{definition}[Vertex-coloring]
An \emph{$r$-coloring} of a graph $G=(V,E)$  is a mapping $c:V\to \{1,\ldots,r\}$
such that:
\begin{itemize}
	\item $c(u)\neq c(v)$ for every $\{u,v\}\in E$.
\end{itemize}
The  {\em chromatic number} of $G$, denoted by $\chi(G)$, is the smallest $r$
such that $G$ has an $r$-coloring.
\end{definition}

We consider the following problem.

\begin{desctight}
\item[Name] Chromatic Number (CN)
\item[Instance] A graph $G=(V,E)$ and a positive integer $r \leq |V|$.
\item[Question] Is there an $r$-coloring for $G$?
\end{desctight}

If $r$ is a constant and not part of the input, the corresponding problem
is denoted by $r$-Chromatic Number ($\CN_{r}$).
Even on 4-regular planar graphs $\CN_{3}$ is NP-complete \cite{Dai80}.

It is well known that bipartite graphs are exactly the
graphs which allow a 2-coloring and that planar graphs are graphs that
allow a 4-coloring.
On undirected co-graphs, the graph coloring problem can be solved
in linear time \cite{CLS81}.

%%%%%%%%%%%%%%%%%%%%%%%%%%%%%%%%%%%%%%%%%%%%%%%%%%%%%%%%%%%%%%%%%%%%%%%%%%
\subsection{Oriented vertex-colorings}
%%%%%%%%%%%%%%%%%%%%%%%%%%%%%%%%%%%%%%%%%%%%%%%%%%%%%%%%%%%%%%%%%%%%%%%%%%

In 1994 Courcelle \cite{Cou94} introduced oriented graph coloring,
which considers only oriented graphs.

\begin{definition}[Oriented vertex-coloring \cite{Cou94}]\label{def-oc}
An \emph{oriented $r$-vertex-coloring} of an oriented graph $G=(V,E)$ is a mapping $c:V\to \{1,\ldots,r\}$
such that:
\begin{itemize}
	\item $c(u)\neq c(v)$ for every $(u,v)\in E$,
	\item $c(u)\neq c(y)$ for every two arcs $(u,v)\in E$ and $(x,y)\in E$ with $c(v)=c(x)$.
\end{itemize}
The {\em oriented chromatic number} of $G$, denoted by $\chi_o(G)$, is the smallest $r$
such that there exists an oriented $r$-vertex-coloring for $G$.
Then $V_i=\{v\in V\mid c(v)=i\}$, $1\leq i\leq r$ is a partition of $V$, which we call {\em color classes}.
\end{definition}

%Within our generalization of oriented colorings in the next section we
%will use the following equivalent characterization of oriented graph coloring from \cite{Sop97}.
For two  digraphs $G_1=(V_1,E_1)$ and $G_2=(V_2,E_2)$, a {\em homomorphism}
from $G_1$ to $G_2$ is a mapping $h: V_1 \to V_2$, which preserves the edges,
i.e., $(u,v) \in E_1$ implies $(h(u),h(v)) \in E_2$.

A homomorphism from $G_1$ to $G_2$ can be regarded as an oriented coloring of $G_1$
in which the vertices of $G_2$ can be seen as color classes. Thus, we call $G_2$ the {\em color graph} of $G_1$.
This leads to equivalent definitions for the
oriented coloring and the oriented chromatic number.
There is an oriented $r$-vertex-coloring of an oriented graph $G_1$
if and only if there is a homomorphism from $G_1$ to some oriented graph $G_2$ on $r$ vertices.
Then, the oriented chromatic number of $G_1$ is the minimum number of vertices in an
oriented graph $G_2$ such that there is a homomorphism from $G_1$ to $G_2$.
Clearly, it is possible to choose $G_2$ as a tournament.

\begin{observation}[\cite{GKL20}]\label{obs-low}
For every oriented graph $G$ it holds that
$\chi(\un(G))\leq \chi_o(G)$.
\end{observation}

However, it is not possible to bound the oriented chromatic number
of an oriented graph $G$ by a function of the (undirected)
chromatic number of $\un(G)$. This has been shown in \cite[Section 3]{Sop16}
by an orientation $K'_{n,n}$ of a $K_{n,n}$ satisfying $\chi_o(K'_{n,n})=2n$ but
$\chi(\un(K'_{n,n}))=2$.

%\begin{lemma}\label{le-isubdigraph}
%Let $G$ be an oriented graph and $H$ be a subdigraph
%of $G$. Then, it holds that $\chi_o(H)\leq \chi_o(G)$.
%\end{lemma}

We now introduce the Oriented Chromatic Number problem.

\begin{desctight}
\item[Name]      Oriented Chromatic Number ($\OCN$)
\item[Given]     An oriented graph $G=(V,E)$ and a positive integer $r \leq |V|$.
\item[Question]  Is there an oriented $r$-vertex-coloring for $G$?
\end{desctight}

\medskip
If $r$ is not part of the input but a constant, we call the related problem the $r$-Oriented Chromatic Number ($\OCN_{r}$).
If  $r\leq 3$, we can decide $\OCN_{r}$ in polynomial time, but still $\OCN_{4}$ is NP-complete \cite{KMG04}. 
Moreover, $\OCN_{4}$ is NP-complete for several restricted
classes of digraphs, e.g., for DAGs \cite{CD06}, line digraphs \cite{OPS08}, and bipartite planar
digraphs with large girth \cite{GO15}.

%ergänze recent result ....

On the other hand,  for every class of graphs of bounded  directed clique-width and every integer $r$ 
the $r$-Oriented Chromatic Number problem can be solved in polynomial time \cite{GKL21a}.
Further, for every oriented co-graph the  Oriented Chromatic Number problem can  be  solved in linear time \cite{GKR19d}.
The latter result even holds for the
super class of all transitive acyclic graphs \cite{GKL20,GKL21a}.
Moreover, for the class of minimal series-parallel digraphs the  Oriented Chromatic Number problem 
can  be  solved in linear time \cite{GKL20,GKL21a,GKL21}.

The definition of oriented vertex-coloring was often
used for undirected graphs, where the
maximum value $\chi_o(G')$ of all possible orientations $G'$ of a graph
$G$ is considered. This leads to the fact that every tree has oriented chromatic number at most $3$.
There are also bounds on the oriented chromatic number for other graph classes, e.g. for outerplanar graphs \cite{Sop97} and
Halin graphs \cite{DS14}. Moreover, the oriented chromatic number of planar graphs
with  large girth was intensively investigated  e.g.\ in \cite{Mar13,Mar15,OP14}.

In \cite{GKL21} we introduced the concept of $g$-oriented $r$-colorings which
generalizes both oriented colorings and colorings of the underlying undirected graph.

\medskip
Next we give an equivalent characterization for  $\OCN$ in terms of a binary integer program.

\begin{remark}\label{rem-lp-ocn}
To formulate  $\OCN$ for some oriented graph $G=(V,E)$ on $n$ vertices as a binary integer program,
we introduce a binary variable $y_j\in\{0,1\}$, $j\in \{1,\ldots, n\}$ such that
$y_j=1$ if and only if color $j$ is used. Further, we use $n^2$ variables
$x_{i,j}\in\{0,1\}$, $i,j\in \{1,\ldots,n\}$ such that
$x_{i,j}=1$ if and only if vertex $v_i$ receives color $j$. The main
idea is to ensure the two conditions of Definition \ref{def-oc}
within conditions (\ref{01-p3-cn0x}) and (\ref{01-px}). W.l.o.g.\ we assume that  $E\neq \emptyset$.
\begin{eqnarray}
\text{Minimize}   \sum_{i=1}^{n}   y_i  \label{01-p1-cn0a}
\end{eqnarray}
subject to
\begin{eqnarray}
\sum_{j=1}^{n} x_{i,j}  & =  &1 \text{ for every }    i \in  \{1,\ldots,n\} \label{01-p2-cn0a} \\
x_{i_0,j} + x_{i_1,j}& \leq   &y_j \text{ for every }   (v_{i_0},v_{i_1})\in E,~j \in \{1,\ldots,n\}  \label{01-p3-cn0x} \\
   \bigvee_{j=1}^n x_{{i_0},j} \wedge x_{i_3,j}             & \leq & 1- \bigvee_{j=1}^n x_{i_1,j} \wedge x_{i_2,j}                            \text{ for every }   (v_{i_0},v_{i_1}),(v_{i_2},v_{i_3})\in E \label{01-px}   \\
y_j   &\in &  \{0,1\} \text{ for every } j \in   \{1,\ldots,n\} \label{01-p4-cna0a} \\
x_{i.j}   &\in &  \{0,1\} \text{ for every } i,j \in  \{1,\ldots,n\}  \label{01-p5-cn0a}
\end{eqnarray}

Equations (\ref{01-px})  are not in propositional logic. In order to reformulate
them for binary integer programming, one can use the results of \cite{Gur14}.
\end{remark}

%%%%%%%%%%%%%%%%%%%%%%%%%%%%%%%%%%%%%%%%%%%%%%%%%%%%%%%%%%%%%%%%%%%%%%%%%%
\subsection{Oriented arc-colorings}
%%%%%%%%%%%%%%%%%%%%%%%%%%%%%%%%%%%%%%%%%%%%%%%%%%%%%%%%%%%%%%%%%%%%%%%%%%

We now define oriented arc colorings for oriented graphs, which were introduced in  \cite{OPS08}.

\begin{definition}[Oriented arc-coloring \cite{OPS08}]\label{def-oac}
An \emph{oriented $r$-arc-coloring} of an oriented graph $G=(V,E)$ is a mapping $c:E\to \{1,\ldots,r\}$ such that:
\begin{itemize}
    \item $c((u,v))\neq c((v,w))$ for every two arcs $(u,v)\in E$ and $(v,w)\in E$
	\item $c((u,v))\neq c((y,z))$ for every four arcs $(u,v)\in E$, $(v,w)\in E$, $(x,y)\in E$, and $(y,z)\in E$, with $c((v,w))=c((x,y))$.
\end{itemize}
The {\em oriented chromatic index} of $G$, denoted with $\chi'_o(G)$, is the smallest $r$
such that $G$ has an oriented $r$-arc-coloring.
Then $E_i=\{e\in E\mid c(e)=i\}$, $1\leq i\leq r$ is a partition of $E$, which we call {\em color classes}.
\end{definition}

There is an oriented $r$-arc-coloring of an oriented graph $G_1$
if and only if there is a homomorphism from line digraph $LD(G_1)$ to some oriented graph $G_2$ on $r$ vertices.
Then, the oriented chromatic index of $G_1$ is the minimum number of vertices in an
oriented graph $G_2$ such that there is a homomorphism from line digraph $LD(G_1)$ to $G_2$.

We consider the following problem.

\begin{desctight}
\item[Name]      Oriented Chromatic Index ($\OCI$)
\item[Given]     An oriented graph $G=(V,E)$ and a positive integer $r \leq |V|$.
\item[Question]  Is there an oriented $r$-arc-coloring for $G$?
\end{desctight}

\medskip
If $r$ is not part of the input but a constant, we call the related
problem the $r$-Oriented Chromatic Index ($\OCI_{r}$).
If  $r\leq 3$, then we can decide $\OCI_{r}$ in polynomial time,
but $\OCI_{4}$ is NP-complete \cite{OPS08}.

The definition of oriented arc-coloring was often
used for undirected graphs, where the
maximum value $\chi'_o(G')$ of all possible orientations $G'$ of a graph
$G$ is considered.
There are  bounds on the oriented chromatic index for special graph classes,
e.g.\ for planar graphs \cite{OPS08} and outerplanar graphs \cite{PS06}.

\begin{observation}[\cite{OPS08}]\label{obs-nu-in}
Let $G$ be an oriented graph. Then, it holds that
$\chi'_o(G) = \chi_o(LD(G))$.
\end{observation}

\begin{observation}[\cite{OPS08}]\label{obs-ind-chr}
Let $G$ be an oriented graph. Then, it holds that
$\chi'_o(G)\leq \chi_o(G)$.
\end{observation}

\medskip
We present an equivalent characterizations for  $\OCI$
using binary integer programs.

\begin{remark}\label{rem-lp-oci}
To formulate  $\OCI$ for some oriented graph $G=(V,E)$ on $n$ vertices and $m$ edges as a binary integer program,
we introduce a binary variable $y_k\in\{0,1\}$, $k\in \{1,\ldots, n\}$, such that
$y_k=1$ if and only if color $k$ is used.\footnote{By Observation \ref{obs-ind-chr} we need at most $n$ colors.}
Further, we use $m\cdot n \leq n^3$ variables
$x_{i,j,k}\in\{0,1\}$, $i,j,k\in \{1,\ldots,n\}$, such that
$x_{i,j,k}=1$ if and only if edge $(v_i,v_j)$ receives color  $k$. The main
idea is to ensure the two conditions of Definition \ref{def-oac}
within conditions (\ref{01-p3-oci}) and (\ref{01-pxa}).
W.l.o.g.\ we assume that $E$ has at least two arcs belonging to a directed path of length two.
\begin{eqnarray}
\text{Minimize}   \sum_{k=1}^{n}   y_k  \label{z2}
\end{eqnarray}
subject to
\begin{eqnarray}
\sum_{k=1}^{n} x_{i,j,k}  & =  &1 \text{ for every }    (v_i,v_j) \in E \label{01-p2-oci} \\
x_{{i_0},i_1,k} + x_{i_1,i_2,k}& \leq   &y_k \text{ for every }   (v_{i_0},v_{i_1}),(v_{i_1},v_{i_2})\in E,~k \in \{1,\ldots,n\}  \label{01-p3-oci} \\
\bigvee_{k=1}^n x_{i_1,i_2,k} \wedge x_{i_3,i_4,k}             & \leq & 1- \bigvee_{k=1}^n x_{i_0,i_1,k} \wedge x_{i_4,i_5,k}                            \text{ for every }   \\
 && (v_{i_0},v_{i_1}),(v_{i_1},v_{i_2}), (v_{i_3},v_{i_4}),(v_{i_4},v_{i_5})\in E \label{01-pxa}   \\
y_k   &\in &  \{0,1\} \text{ for every } k \in   \{1,\ldots,n\} \label{01-p4-oci} \\
x_{i,j,k}   &\in &  \{0,1\} \text{ for every } i,j,k \in  \{1,\ldots,n\}  \label{01-p5-oci}
\end{eqnarray}

Equations (\ref{01-pxa})  are not in propositional logic. In order to reformulate
them for binary integer programming, one can use the results of \cite{Gur14}.
\end{remark}

%%%%%%%%%%%%%%%%%%%%%%%%%%%%%%%%%%%%%%%%%%%%%%%%%%%%%%
\section{Edge Series-Parallel Digraphs}
%%%%%%%%%%%%%%%%%%%%%%%%%%%%%%%%%%%%%%%%%%%%%%%%%%%%%%

Before we consider edge series-parallel digraphs we recall some results for the well-known
undirected class of series-parallel graphs.
%
%%%%%%%%%%%%%%%%%%%%%%%%%%%%%%%%%%%%%%%%%%%%%%%%%%%%%%
%\subsection{Series-Parallel Graphs}
%%%%%%%%%%%%%%%%%%%%%%%%%%%%%%%%%%%%%%%%%%%%%%%%%%%%%%
%
%
Undirected series-parallel graphs are graphs with two distinguished vertices called terminals,
formed recursively by parallel and series composition \cite[Section 11.2]{BLS99}.
These graphs are interesting from a practical point of view due their applications
in  modeling series and parallel electric circuits. Furthermore, they also
play an important role in theoretical computer science, since they
have tree-width at most 2 and are $K_4$-minor free graphs \cite{Bod98}.
%Series-parallel graphs can be used to model series and parallel electric circuits.

The chromatic number of series-parallel graphs can easily be bounded as follows.

\begin{proposition}[\cite{Sey90}]
Let $G$ be some series-parallel graph. Then, it holds that
$\chi(G)\leq 3$.
\end{proposition}

%The  class of undirected series-parallel  graphs was considered in  \cite{Sop97}
%by showing that every orientation of a series-parallel  graph has oriented
%chromatic number at most 7. This bound can not be applied to  minimal series-parallel digraphs,
%since  the set of all $\overrightarrow{K_{n,m}}$  is a subset of minimal series-parallel digraphs
%and the underlying graphs are even of unbounded tree-width and thus, no series-parallel  graphs.

The oriented chromatic number of undirected series-parallel graphs was
considered in \cite{Sop97}.

\begin{theorem}[\cite{Sop97}]\label{sop-bd7}
Let $G'$ be some orientation of a series-parallel graph $G$. Then, it holds that
$\chi_o(G')\leq 7$.
\end{theorem}

In \cite{Sop97} it was also shown that this bound is tight.
In \cite{PS06} this was strengthened by giving a triangle-free orientation of
a series-parallel graph of order 15 and oriented chromatic number $7$.

For the chromatic index of orientations of undirected series-parallel graphs
Observation \ref{obs-ind-chr} and Theorem
\ref{sop-bd7} lead to the following bound.

\begin{corollary}\label{chri-bd7}
Let $G'$ be some orientation of a series-parallel graph $G$. Then, it holds that
$\chi'_o(G')\leq 7$.
\end{corollary}

In \cite{PS06} it was shown that the bound is tight (even for an orientation
of an outerplanar graph).

We recall the definition of edge series-parallel digraphs, originally defined as edge series-parallel multidigraphs, from \cite{VTL82}.

\begin{definition}[Edge Series-Parallel Multidigraphs]\label{def-esp}
The class of {\em edge series-parallel multidigraphs}, {\em esp-digraphs} for
short, is recursively defined as follows.
\begin{enumerate}[(i)]

\item  Every digraph
of two distinct vertices joined by a single arc $(\{u,v\},\{(u,v)\})$,
denoted by $(u,v)$, is an {\em edge series-parallel multidigraph}.

\item If  $G_1=(V_1,A_1)$ and $G_2=(V_2,A_2)$
are vertex-disjoint minimal edge series-parallel multidigraphs, then
\begin{enumerate}

\item the {\em parallel composition} $G_1\cup G_2$, which
identifies the source of $G_1$ with the
source of $G_2$ and the sink of $G_1$ with the sink of $G_2$,
is an {\em edge series-parallel multidigraph} and

\item the {\em series composition} $G_1\times G_2$, which
identifies the sink of $G_1$ with the source of $G_2$,
is an {\em edge series-parallel multidigraph}.
\end{enumerate}
\end{enumerate}
\end{definition}

An expression $X$ using  the operations of Definition \ref{def-esp}
is called an {\em esp-expression} and  $\g(X)$ the defined graph. For a better understanding we now give an example of such an expression.

\begin{example} \label{ex-esp}
The  esp-expression
$$
X_1=  \left((v_1,v_2)\times \left(\left((v_2,v_3)\times \left((v_3,v_4)\times (v_4,v_5) \right)\right)\cup (v_2,v_5)\right)\right) \times (v_5,v_6)
$$
defines the esp-digraph
shown in Figure \ref{F02e}.
\end{example}

\begin{figure}[hbtp]
\centerline{\includegraphics[width=0.55\textwidth]{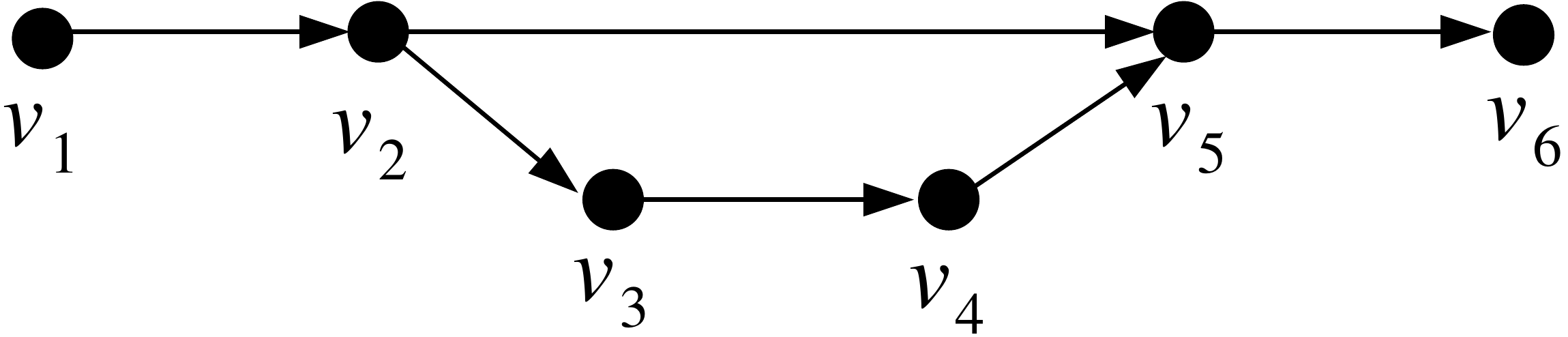}}
\caption{$\gr(X_1)$ in Example \ref{ex-esp}.}
\label{F02e}
\end{figure}

Several classes of digraphs are included in the set of all esp-digraphs.

\begin{example}\label{ex-2-esp}
\begin{enumerate}
\item
Every oriented path on $n$ vertices is an esp-digraph by the following esp-expression.
$$X_{P_n}=(\ldots(((v_1,v_2)\times(v_2,v_3))\times(v_3,v_4))  \ldots ) \times (v_{n-1},v_n)$$

\item
Every oriented cycle on $n\geq 3$ vertices with one reversed arc is an esp-digraph by the following esp-expression.
$$X_{C'_n}= X_{P_n} \cup (v_1,v_n)$$
\end{enumerate}
\end{example}

For every esp-digraph we can define a tree structure,
denoted as {\em esp-tree}. The leaves of the esp-tree represent the
arcs of the digraph and the inner nodes of the esp-tree  correspond
to the operations applied on the sub-expressions defined by the subtrees.
For some vertex $u$ of esp-tree $T$ we denote by $T(u)$
the subtree rooted at $u$ and by $X(u)$ the {\em sub-expression}
defined by $T(u)$.

For every esp-digraph one can construct an esp-tree in linear time
\cite{Val78}.

%For every esp-digraph we can define a tree structure,
%denoted as {\em esp-tree}. The leaves of the esp-tree represent the
%arcs of the graph and the inner nodes of the esp-tree  correspond
%to the operations applied on the sub-expressions defined by the subtrees.
%For every esp-digraph
%one can construct an esp-tree in liner time
%\cite{Val78}.

In \cite{HY87} the notation two-terminal series-parallel (TTSP) graphs is used
for the same graphs and give a parallel algorithm for recognizing directed
series-parallel graphs. Further, \cite{Epp92} gives an improved parallel
algorithm for recognizing directed (and undirected) series-parallel graphs.

\begin{observation}
Let $G$ be an esp-digraph. Then, it holds that $G$ has exactly
one source and exactly one sink.
\end{observation}

%Several classes of digraphs are included in the set of all esp-digraphs.

%\begin{proposition}
%Every  oriented path  is an esp-digraph.
%\end{proposition}

For every digraph $G=(\{v,u\},(u,v))$, $\un(G)$ is series-parallel graph. 
Further, we can replace every parallel composition by an parallel composition in the undirected case, and every series composition by a series composition in the undirected case, which leads to the following result.

\begin{proposition} \label{esp-un}
Let $G$ be an esp-digraph. Then, it holds that
$\un(G)$ is a series-parallel graph.
\end{proposition}

\begin{remark}\label{rem-esp-ori}
By Proposition \ref{esp-un} every esp-digraph is an orientation of a series-parallel graph.
\end{remark}

%%%%%%%%%%%%%%%%%%%%%%%%%%%%%%%%%%%%%%%%%%%%%%%%%%%%%%%%%%%%%%%%%%%%%%%%%%
\subsection{Oriented Arc-Colorings of Edge Series-Parallel Digraphs}
%%%%%%%%%%%%%%%%%%%%%%%%%%%%%%%%%%%%%%%%%%%%%%%%%%%%%%%%%%%%%%%%%%%%%%%%%%

Since every esp-digraph is an orientation of
a series-parallel graph by Corollary \ref{chri-bd7} we have the following bound.
%The following bound was mentioned in \cite{PS06} for orientations of
%series-parallel graphs.

\begin{corollary}\label{cor-77}
Let $G$ be an esp-digraph. Then, it holds that
$\chi'_o(G)\leq 7$.
\end{corollary}

Alternatively, the last result
can be obtained from Proposition \ref{prop-msp-7}, Lemma \ref{le-vtl} and
Observation \ref{obs-nu-in}.

% updaten

\begin{remark}
\label{rem-esp-chiprime}
We can also bound the oriented chromatic index of an esp-digraph $G$  using
the corresponding line digraph $LD(G)$ which is an msp-digraph (cf.\ Definition \ref{def-msp}) by
Lemma \ref{le-vtl}.
%there is some esp-digraph $G'$, such that $G=LD(G')$ and
%by Observation \ref{obs-nu-in} we know $\chi'_o(G') = \chi_o(LD(G'))$, i.e.
%
%
%
$$\begin{array}{lclll}
\chi'_o(G)&=&  \chi_o(LD(G))   &\text{Observation } \ref{obs-nu-in} \\
      &\leq & 7   &   \text{Lemma } \ref{le-vtl}  \text{ and Proposition } \ref{prop-msp-7}
\end{array}
$$
\end{remark}

%

%By Lemma \ref{le-vtl} and Observation \ref{obs-nu-in} an oriented arc-coloring of
%an edge series-parallel multidigraph can be
%obtained by an oriented vertex coloring of  minimal
%vertex series-parallel digraph, which we studied in \cite{GKL20,GKL21}.

The results of  \cite{PS06} even show that $7$ is a tight upper bound for the oriented
chromatic index of every orientation of series-parallel graphs (even for an orientation
of an outerplanar graph).

In order  to show that this bound is also tight for the subclass of esp-digraphs we give the
next example.

% in PS06 (Fig 3) wird eine orientierung eines aussenplanaren graphen als Bsp
% angegeben aber ob das auch ein series-paralleler gr

\begin{example}\label{ex-six2}
The esp-expression
$$
\begin{array}{lcl}
X_2 &=& (v_{1},v_2) \times  ((v_2,v_5) \cup (v_2,v_3) \times ((v_3,v_5) \cup (v_3,v_4) \times (v_4,v_5)))\times \\

&&  ((v_5,v_9) \cup ((v_5,v_7) \cup (v_5,v_6) \times (v_{6},v_7))\times ((v_{7},v_9) \cup
(v_{7},v_8) \times (v_{8},v_9)))\times \\
&&((v_9,v_{16})\cup ((v_9,v_{13}) \cup ((v_9,v_{11}) \cup (v_9,v_{10}) \times (v_{10},v_{11}))\times \\
&&((v_{11},v_{13}) \cup (v_{11},v_{12}) \times (v_{12},v_{13})))\times  \\
&&((v_{13},v_{16}) \cup((v_{13},v_{15}) \cup (v_{13},v_{14}) \times (v_{14},v_{15}))\times (v_{15},v_{16})) )\times  (v_{16},v_{17})
\end{array}
$$
defines the esp-digraph on 17 vertices shown in Figure \ref{F09c}. Further, by Observation \ref{obs-nu-in}
and since $\g(X_5)$,where $X_5$ is defined in Example \ref{ex-sixa}, is the line digraph of $\g(X_2)$
it holds that
$$\chi'_o(\g(X_2))=\chi_o(LD(\g(X_2)))= \chi_o(\g(X_5))= 7.$$
This implies that the bound of Corollary \ref{cor-77} is best possible.
\end{example}

\begin{figure}[hbtp]
\centerline{\includegraphics[width=0.9\textwidth]{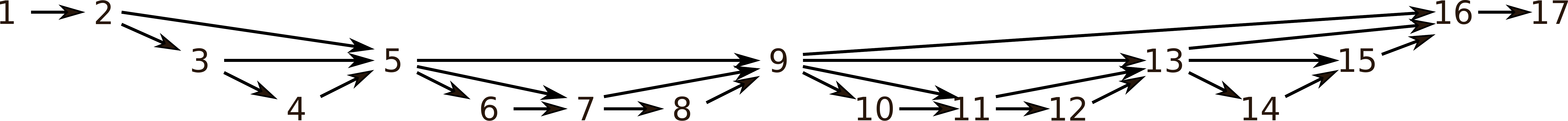}}
\caption{$\gr(X_2)$ in Example \ref{ex-six2}.}
\label{F09c}
\end{figure}

By Theorem \ref{msp-ori-c} and Observation \ref{obs-nu-in} we obtain the following result.

\begin{corollary}\label{esp-index-com}
Let $G$ be an esp-digraph. Then, the oriented
chromatic index of $G$ can be computed in linear time.
\end{corollary}

%%%%%%%%%%%%%%%%%%%%%%%%%%%%%%%%%%%%%%%%%%%%%%%%%%%%%%%%%%%%%%%%%%%%%%%%%%
\subsection{Oriented Vertex-Colorings of Edge Series-Parallel Digraphs}
%%%%%%%%%%%%%%%%%%%%%%%%%%%%%%%%%%%%%%%%%%%%%%%%%%%%%%%%%%%%%%%%%%%%%%%%%%

%By Proposition \ref{esp-un} and Theorem \ref{sop-bd7} shown in \cite{Sop97} we conclude
%that the oriented chromatic number of esp-digraphs
%is at most $7$.

Since every esp-digraph is an orientation of
a series-parallel graph by Theorem \ref{sop-bd7} we have the following bound.

%The following bound was mentioned in \cite{PS06} for orientations of
%series-parallel graphs.

\begin{corollary}\label{cor-78}
Let $G$ be an esp-digraph. Then, it holds that
$\chi_o(G)\leq 7$.
\end{corollary}

The proof of Theorem \ref{sop-bd7} given in \cite{Sop97} uses the
color graph $QR_7=(V,E)$ where $V=\{1,2,3,4,5,6,7\}$ and
$E=\{(i,j)\mid j-i  \equiv 1, 2, \text{ or } 4 ~(\bmod ~ 7) \}$
which is built from the
non-zero quadratic residues of 7 and is shown in Figure \ref{F06}.
We next give an alternative proof of Corollary \ref{cor-78} using the
recursive structure of esp-digraphs.

\begin{figure}[hbtp]
\centerline{\includegraphics[width=0.28\textwidth]{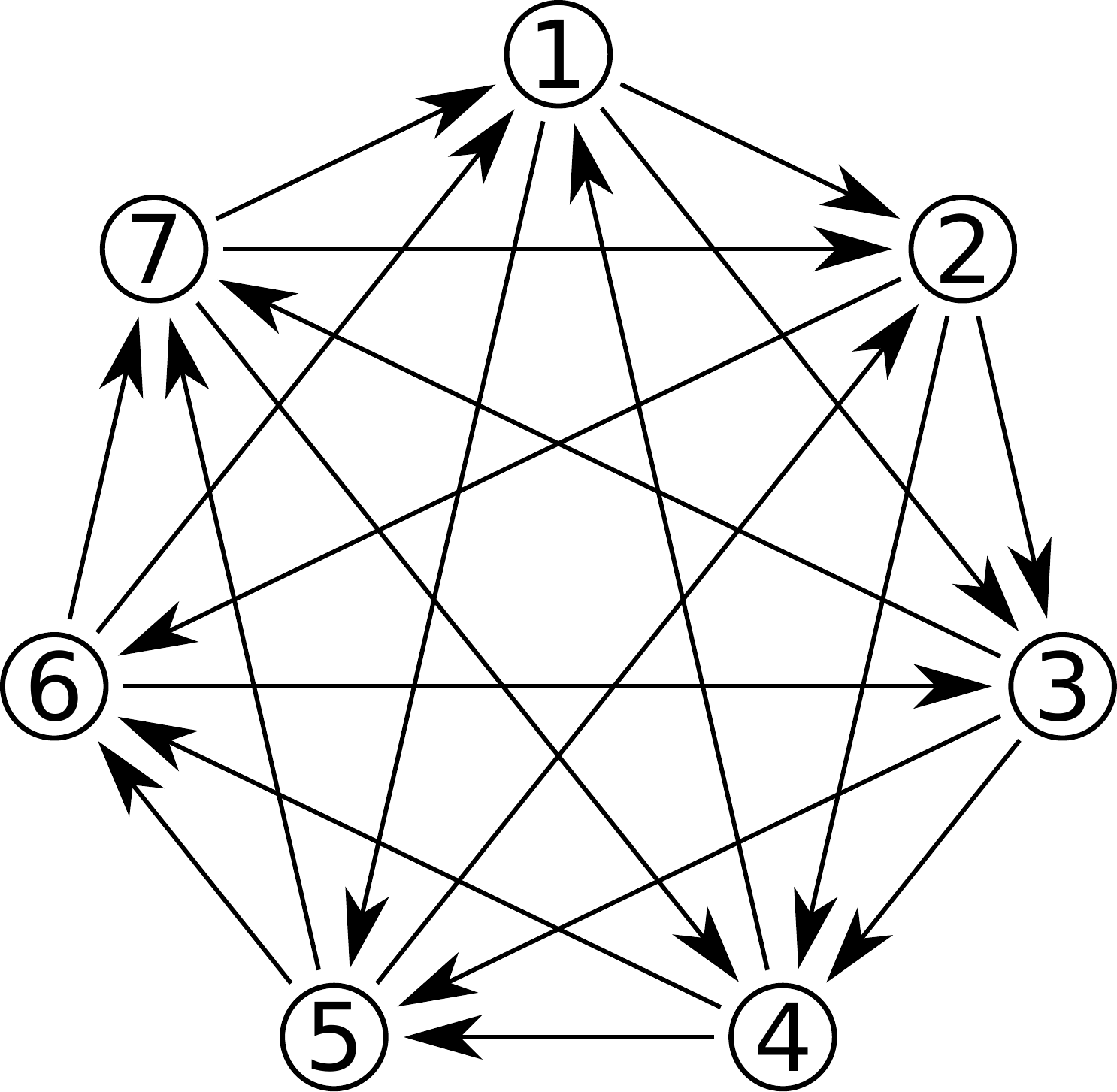}}
\caption{Color graph used in the proof of Theorem  \ref{esp-bd7}.}
\label{F06}
\end{figure}

\begin{theorem}\label{esp-bd7}
Let $G$ be an esp-digraph. Then, it holds that $\chi_o(G)\leq 7$.
\end{theorem}

\begin{proof}
Let $G=(V_G,E_G)$ be some series-parallel digraph. We use the color graph $H=(\{1,2,3,4,5,6,7\},E_H)$
shown in Figure \ref{F06} to define an oriented $7$-vertex-coloring $c:V_G\to \{1,\ldots,7\}$ for $G$.

First, we color the source of $G$ by $1$ and the sink of $G$ by $2$.
Next, we recursively decompose $G$ in order to color all vertices of $G$.
In any step we will keep the invariant that $(c(q),c(s))\in E_H$, if
$q$ is the source of $G$ and $s$ is the sink of $G$.

\begin{itemize}
\item If $G$ emerges from parallel composition $G_1 \cup G_2$, we proceed
with coloring $G_1$ and $G_2$ on its own. Doing so, the color of the source
and sink in $G_1$ and $G_2$ will not be changed.

\item If $G$ emerges from series composition $G_1 \times G_2$, let $a$ be
the color of the source and $c$ be the color of the sink in $G$. For
the sink of $G_1$ and the source of $G_2$ we choose color $b$, such that
the arcs $(a,b)$ and $(b,c)$ are in color graph $H$. This is always possible
by the three possible cases shown in Figure \ref{F05}. For every (red) arc there
is a path of length two (blue) with the same start and end vertex.

\begin{figure}[hbtp]
\centerline{\includegraphics[width=0.99\textwidth]{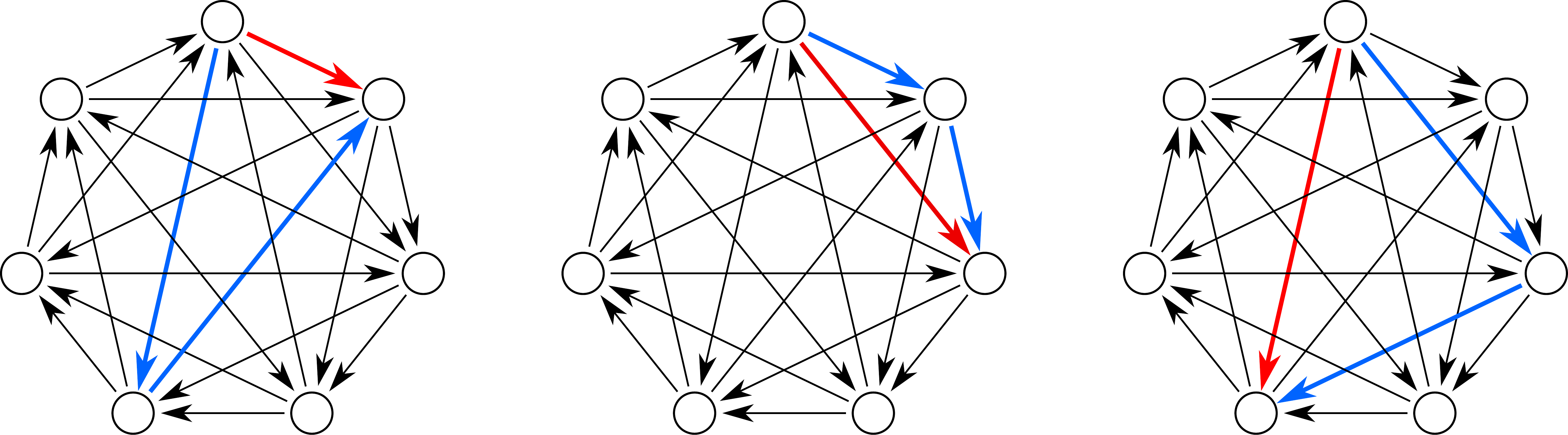}}
\caption{Three possible cases for arcs in the color graph used in the proof of Theorem  \ref{esp-bd7}.}
\label{F05}
\end{figure}

\item If $G$ consists of a pair of vertices connected by a single arc,
the coloring is given by our invariant.
\end{itemize}
This shows the statement of the theorem.
\end{proof}

For the optimality of the shown bound, we give the following
example.

\begin{example}\label{ex-six}
The esp-expression
$$
\begin{array}{lcl}
X_3  &=& ((v_1,v_4)\cup((v_1,v_2)\times((v_2,v_4)\cup ((v_2,v_3)\times (v_3,v_4)))))   \times \\
   & & ((((v_4,v_6)\cup((v_4,v_5)\times (v_5,v_6))) \times (v_6,v_7))\cup (v_4,v_7))
\end{array}
$$
defines the esp-digraph on 7 vertices shown in Figure \ref{F09} and it
obviously holds that $\chi_o(\g(X_3))=7$.
This implies that the bound of Theorem \ref{esp-bd7} is best possible.
\end{example}

\begin{figure}[hbtp]
\centerline{\includegraphics[width=0.4\textwidth]{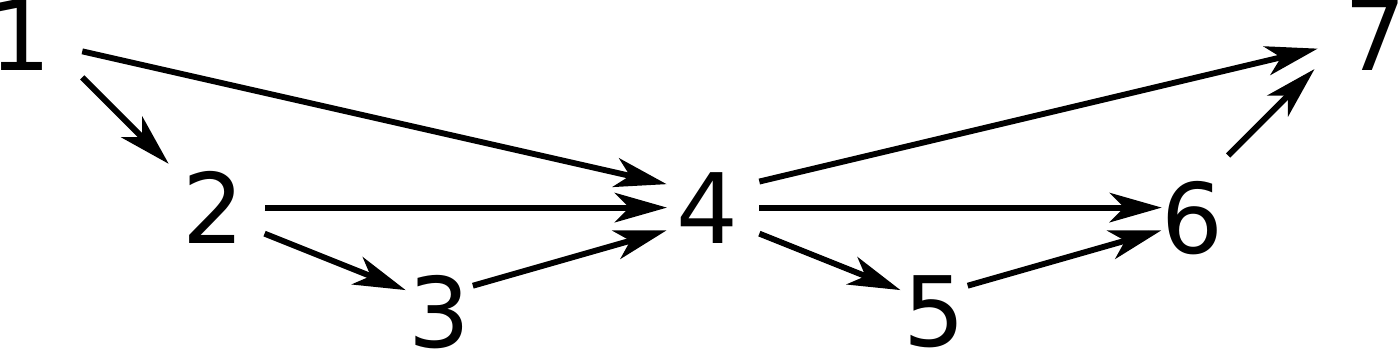}}
\caption{$\gr(X_3)$ in Example   \ref{ex-six}.}
\label{F09}
\end{figure}

In order to compute the  oriented
chromatic number of an esp-digraph $G=(V,E)$ defined by an esp-expression $X$,
we recursively compute the set $F(X)$ of all triples $(H,\ell,r)$ such that
$H$ is a color graph for $G$, where $\ell$ and $r$ are the
colors of the source and sink, respectively, in $G$ with respect
to the coloring by $H$.
The number of vertex labeled, i.e., the vertices are distinguishable
from each other, oriented graphs  on $n$
vertices is  $3^{\nicefrac{n(n-1)}{2}}$.
By Theorem \ref{esp-bd7} and also by Corollary \ref{cor-78} we can conclude that
$$|F(X)|\leq 3^{\nicefrac{7(7-1)}{2}}\cdot 7 \cdot 7\in \bigo(1)$$
which is independent of the size of $G$.

For two color graphs $H_1=(V_1,E_1)$ and  $H_2=(V_2,E_2)$
we define $H_1+H_2=(V_1\cup V_2,E_1\cup E_2)$.

\begin{lemma}
%[$\bigstar$\footnote{The proofs of the results marked with a $\bigstar$
%are omitted due to space restrictions.}]
\label{le1}
%Let $G$ be an esp-digraph. Then
\begin{enumerate}
\item
For every $(u,v)\in E$ it holds  $$F((u,v))=\{((\{i,j\},\{(i,j)\} ),i,j) \mid 1\leq i,j \leq 7,~ i\neq j\}.$$

\item
For every two esp-expressions $X_1$ and  $X_2$ we obtain $F(X_1\cup X_2)$
from  $F(X_1)$ and $F(X_2)$ as follows. For every $(H_1,\ell_1,r_1) \in F(X_1)$ and
every $(H_2,\ell_2,r_2) \in F(X_2)$ such that graph $H_1+H_2$ is oriented, $\ell_1=\ell_2$,
and $r_1=r_2$, we put $(H_1 + H_2,\ell_1,r_1)$ into $F(X_1\cup X_2)$.

\item
For every two esp-expressions $X_1$ and  $X_2$ we obtain $F(X_1\times X_2)$
from  $F(X_1)$ and $F(X_2)$ as follows. For every $(H_1,\ell_1,r_1) \in F(X_1)$ and
every $(H_2,\ell_2,r_2) \in F(X_2)$ such that graph $H_1+H_2$ is oriented,  and $r_1=\ell_2$, we
put $((V_1\cup V_2,E_1\cup E_2),\ell_1 ,r_2)$  into $F(X_1\times X_2)$.
\end{enumerate}
\end{lemma}

\begin{proof}
We show for each operation that the stated formulas  hold.
\begin{enumerate}

\item Set $F((u,v))$ includes obviously all possible solutions to color the end vertices of every arc
on its own with the  $7$ given colors.

\item Set $F(X_1)$ includes all possible solutions for coloring $X_1$, just as $F(X_2)$ for $X_2$. In particular we have solutions included, that are equal but a permutation of the colors. Since the sources and sinks are each identified with each other, we only keep solutions where $\ell_1= \ell_2$ and $r_1=r_2$. In this step it is essential that we kept all possible solutions before, even if they are just permutations of the different colors.
$H_1+H_2$ is oriented and has by construction at most 7 vertices. Since in $\g(X_1 \cup X_2)$ there are no additional edges compared to $E_1\cup E_2$, every vertex can get the same color as in the individual solutions, such that all vertices are legally colored. So $(H_1 + H_2,\ell_1 ,r_1)$ is an possible solution to color and thus $(H_1 + H_2,\ell_1,r_1)\in F(X_1\cup X_2)$.

Let $(H,\ell ,r)\in F(X_1 \cup X_2)$, then we can take an induced subdigraph $H_1$ which colors all the vertices of $\g(X_1)$ as well as $H_2$ which colors all the vertices of $H_2$. Let $\ell_1 = \ell$ be the color of the source in $\g(X_1)$ and $r_1 = r$ the color of the sink in $\g(X_2)$.
It holds that $(H_1,\ell_1,r_1)\in F(X_1)$. The same arguments hold for $X_2$ such that $(H_2,\ell_2,r_2)\in F(X_2)$.

\item Set $F(X_1)$ includes all possible solutions for coloring $X_1$, just as $F(X_2)$ for $X_2$.
In particular we have solutions included, that are equal but a permutation of the colors. Since the source and the sink are identified with each other, we only keep solutions where $r_1=\ell_2$. In this step it is essential that we kept all possible solutions before, even if they are just permutations of the different colors.
$H_1+H_2$ is oriented and has by construction at most 7 vertices. In $\g(X_1\times X_2)$ there no additional edges compared to $E_1\cup E_2$, every vertex can get the same color as in the individual solutions, such that all vertices are legally colored. So $(H_1 + H_2, \ell_1, r_2)$ is a possible solution and thus, $(H_1 + H_2, \ell_1, r_2) \in F(X_1\times X_2)$.

Let $(H,\ell,r)\in F(X_1\times X_2)$, then we can take an induced subdigraph $H_1$ which colors all the vertices of $\g(X_1)$ as well as $H_2$ which colors all the vertices of $H_2$. Let $\ell_1 = \ell$ be the color of the source in $\g(X_1)$ and $r_1$ the color of the sink in $\g(X_2)$.
It holds that $(H_1,\ell_1,r_1)\in F(X_1)$. The same arguments hold for $X_2$, if $\ell_2$ is the color of the source of $\g(X_1)$ and $r_1=r$ is the color of the sink of $\g(X_2)$, such that $(H_2,\ell_2,r_2)\in F(X_2)$.
\end{enumerate}
This shows the statements of the lemma.
\end{proof}

The optimal solution for digraph $G$ given by an esp-expression $X$ is always included in $F(X)$ since all possible sub-solutions are maintained in the process and not only the optimal solutions. We can show this shortly by contradiction. Assumed there exists an optimal solution $(H,\ell, r)$ for $X_1\times X_2$, but $(H,\ell, r)\not\in F(X_1\times X_2)$ such that $(H,\ell, r)$ was not taken into the solution. Thus, for either $X_1$ or $X_2$ (which we call $X_i$ in the following), the solution of coloring the vertices with color graph $H'$, which is an induced subdigraph of $H$ and which only contains the colors we need for coloring $X_i$, was not part of the solution $F(X)$. But since there are all the possible solutions in $F(X)$ and not only minimal solutions, this is a contradiction to our procedure. The same holds for the parallel composition $X_1\cup X_2$. Thus, we find a minimal coloring for $G$.
We know from Theorem \ref{sop-bd7} that the number of colors in an optimal solution is limited by 7.

\begin{corollary}\label{cor1}
There is an oriented vertex $r$-coloring for an esp-digraph $G$
which is given by an esp-expression $X$ if and only if there is some $(H,\ell,r)\in F(X)$
such that color graph $H$ has $r$ vertices.
Therefore, $\chi_o(G)=\min\{|V| \mid ((V,E),\ell,r)\in F(X)\}$.
\end{corollary}

\begin{theorem}\label{esp-num-com}
Let $G$ be an esp-digraph. Then, the oriented
chromatic number of $G$ can be computed in linear time.
\end{theorem}

\begin{proof}
Let $G=(V,E)$ be an esp-digraph with $n=|V|$ vertices and $m=|E|$ edges and let  $T$ be an esp-tree for $G$
with root $r$. For a vertex $u$ of $T$ we denote by $T_u$
the subtree rooted at $u$ and by $X_u$ the esp-expression defined by $T_u$.

For computing the oriented
chromatic number for an esp-digraph $G$, we traverse esp-tree $T$ in bottom-up order.
For every vertex $u$ of $T$ we can compute $F(X_u)$ by
following the rules given in Lemma \ref{le1}. By Corollary \ref{cor1} we can solve our
problem using $F(X_r)=F(X)$.

An esp-tree $T$ can be computed in $\bigo(n+m)$ time from $G$, see \cite{Val78}.
By Lemma \ref{le1} we obtain the following running times.
\begin{itemize}
\item
For every vertex $(u,v)\in E$ set  $F((u,v))$ is computable in $\bigo(1)$ time.

\item
For every two esp-expressions $X_1$ and $X_2$  set
$F(X_1  \cup X_2)$   can be computed in  $\bigo(1)$ time from
$F(X_1)$ and $F(X_2)$.

\item
For every two esp-expressions $X_1$ and $X_2$  set
$F(X_1  \times X_2)$ can be computed in   $\bigo(1)$ time from $F(X_1)$ and $F(X_2)$.
\end{itemize}

Since $T$ consists of $n$ leaves and $n-1$ inner vertices,
the overall running time is in $\bigo(n+m)$. 
\end{proof}

%%%%%%%%%%%%%%%%%%%%%%%%%%%%%%%%%%%%%%%%%%%%%%%%%%%%%%
\section{Minimal Vertex Series-Parallel Digraphs}\label{sol-co}
%%%%%%%%%%%%%%%%%%%%%%%%%%%%%%%%%%%%%%%%%%%%%%%%%%%%%%

We recall the definition of minimal\footnote{In order to motivate the notation of {\em minimal} vertex series-parallel
graphs we refer to a super class of series-parallel digraphs, which are
are exactly the digraphs whose transitive closure equals
the transitive closure of a minimal series-parallel digraph \cite{VTL82}.} vertex
series-parallel digraphs from \cite{BG18} which are
based on \cite{VTL82}.

\begin{definition}[Minimal Vertex Series-Parallel Digraphs]\label{def-msp}
The class of {\em minimal vertex series-parallel digraphs}, {\em msp-digraphs} for
short,  is recursively defined as follows.

\begin{enumerate}[(i)]
\item Every digraph on a single vertex $(\{v\},\emptyset)$,
denoted by $v$, is a {\em minimal vertex series-parallel digraph}.

\item If  $G_1=(V_1,E_1)$ and $G_2=(V_2,E_2)$  are vertex-disjoint minimal
vertex series-parallel digraphs and
$O_1$ is the set of vertex of outdegree $0$ (set of sinks) in $G_1$ and
$I_2$ is the set of vertices of indegree $0$ (set of sources) in $G_2$, then
\begin{enumerate}
\item
the  {\em parallel composition}
$G_1 \cup G_2 =(V_1\cup V_2, E_1\cup E_2)$
is a {\em minimal vertex series-parallel digraph} and

\item
the {\em series composition}
$G_1 \times G_2=(V_1\cup V_2,E_1\cup E_2\cup (O_1 \times I_2))$
is a {\em minimal vertex series-parallel digraph}.
\end{enumerate}
\end{enumerate}
\end{definition}

An expression $X$ using  the  operations of Definition \ref{def-msp}
is called an {\em msp-expression} and  $\g(X)$ the defined graph. We illustrate such an expression with the following example.

\begin{example}\label{ex-msp}
The msp-expression
$$
X_4=  \left(v_1\times \left(\left(v_2\times \left(v_3\times v_4 \right)\right)\cup v_5\right)\right) \times v_6
$$
defines the  msp-digraph shown in Figure \ref{F02}.
\end{example}

\begin{figure}[hbtp]
\centerline{\includegraphics[width=0.5\textwidth]{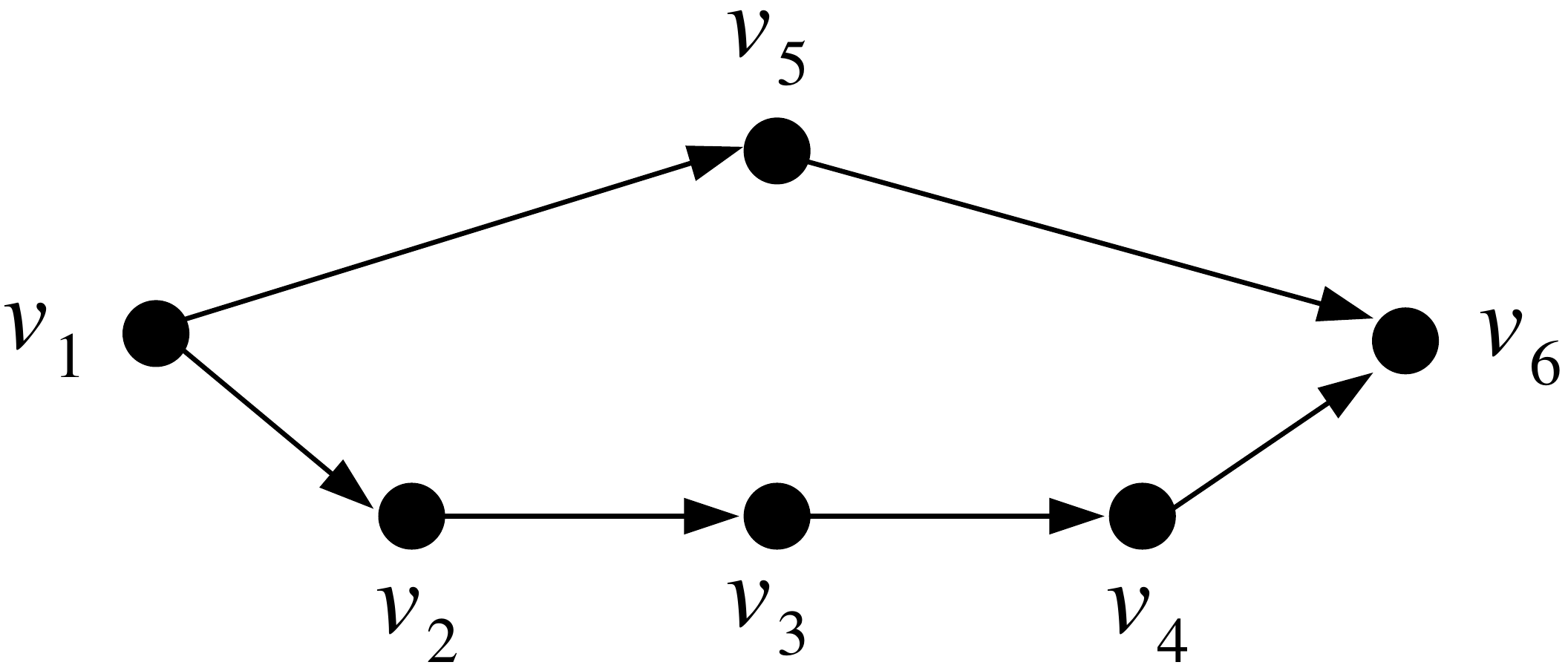}}
\caption{$\gr(X_4)$ in Example \ref{ex-msp}.}
\label{F02}
\end{figure}

Several classes of digraphs are included in the set of all msp-digraphs.

\begin{example}\label{ex-2-msp}
\begin{enumerate}
\item Every oriented bipartite graph $\overrightarrow{K_{n,m}}$ is an msp-digraph by the following msp-expression.
$$X_{K_{n,m}}= (v_1\cup \ldots \cup v_n)\times (w_1\cup \ldots \cup w_m) $$

\item Every  in- and out-rooted tree $T$ is an msp-digraph. An  
msp-expression $X$ can be obtained by inserting the vertices of $T$ into $X$ using a bottom-up order. 
We denote by $X_{v_i}$ the msp-expression for the subtree of $T$ rooted at $v_i$. 
For every leaf $v_i$ of $T$ we obviously have the msp-expression $X_{v_i}=v_i$. 
For every inner vertex $v_i$ with successors $v_{j_1},\ldots,v_{j_i}$ in $T$
the expressions $X_{v_{j_1}},\ldots,X_{v_{j_i}}$ are first combined by 
parallel compositions to sub-expression $X'_{v_i}$ and afterwards $v_i$ is combined with  $X'_{v_i}$ using
a series composition to obtain sub-expression $X_{v_i}$.
\end{enumerate}
\end{example}

For every msp-digraph we can define a tree structure,
denoted as {\em msp-tree}. The
vertices of the graph are represented by the leaves of the msp-tree. Meanwhile, the inner nodes of the msp-tree correspond
to the operations which are applied on the sub-expressions defined by the subtrees.
For a vertex $u$ of msp-tree $T$ we denote by $T(u)$
the subtree  which is rooted at $u$ and by $X(u)$ the {\em sub-expression}
defined by $T(u)$.

For every msp-digraph we can construct a msp-tree in linear time, see
\cite{VTL82}.

Further, there is a close relation between esp-digraphs
and msp-digraphs, which the following Lemma shows.

\begin{lemma}[\cite{VTL82}]\label{le-vtl}
An acyclic multidigraph $G$ with a single source and a single sink is an
esp-digraph if and
only if its line digraph $LD(G)$ is a  msp-digraph.
\end{lemma}

\begin{example}
The msp-digraph $\g(X_4)$, which is defined in Example \ref{ex-msp} and
shown in Figure \ref{F02},
defines the line digraph of esp-digraph $\g(X_1)$ defined in Example \ref{ex-esp}, see Figure \ref{F02e}.
\end{example}

Next we compare the classes of msp-digraphs and esp-digraphs.

\begin{remark}\label{rem-inclu}
The classes of msp-digraphs and esp-digraphs are incomparable in terms of  set inclusion.
Example \ref{ex-2-esp}(2.) gives a class of esp-digraphs which are not msp-digraphs and Example \ref{ex-2-msp}(1.).
gives a class of msp-digraphs which are not esp-digraphs.

\end{remark}

\begin{remark}
In contrast to Remark \ref{rem-esp-ori} on esp-digraphs, 
there are classes of msp-digraphs which are not orientations of a series-parallel graph. This can 
be shown by Example \ref{ex-2-msp}(1.).
The underlying undirected graphs of  oriented bipartite graphs $\overrightarrow{K_{n,m}}$ have unbounded tree-width while the set of series-parallel graphs has  tree-width at most 2 \cite{Bod98}.
\end{remark}

%%%%%%%%%%%%%%%%%%%%%%%%%%%%%%%%%%%%%%%%%%%%%%%%%%%%%%%%%%%%%%%%%%%%%%%%%%
\subsection{Oriented Vertex-Colorings of Minimal Vertex Series-Parallel Digraphs}
%%%%%%%%%%%%%%%%%%%%%%%%%%%%%%%%%%%%%%%%%%%%%%%%%%%%%%%%%%%%%%%%%%%%%%%%%%

Using the recursive structure of msp-digraphs 
we could show the following bound on their oriented chromatic number.

\begin{proposition}[\cite{GKL20,GKL21}]\label{prop-msp-7}
Let $G$ be an msp-digraph. Then, it holds that $\chi_o(G)\leq 7$.
\end{proposition}

\begin{remark}
\label{rem-msp-chi}
We can also bound the oriented chromatic number of a msp-digraph $G$  using
the corresponding root digraph $G'$ which is an esp-digraph 
%(cf.\ Definition \ref{def-esp}) 
by
Lemma \ref{le-vtl}.
%there is some esp-digraph $G'$, such that $G=LD(G')$ and
%by Observation \ref{obs-nu-in} we know $\chi'_o(G') = \chi_o(LD(G'))$, i.e.
%
%
%
$$\begin{array}{lclll}
\chi_o(G)&=    & \chi_o(LD(G')) & \text{Lemma }       \ref{le-vtl}   \\
         &=    & \chi'_o(G')    & \text{Observation } \ref{obs-nu-in} \\
         &\leq & 7              & \text{Corollary }   \ref{cor-77}
\end{array}
$$
\end{remark}

For the optimality of the shown bound, we recall from  \cite{GKL20} the following
example.

\begin{example}[\cite{GKL20}]\label{ex-sixa}We consider the
msp-expression\footnote{In all expression we assume that
the series composition  binds more strongly than the parallel composition.}
$$
\begin{array}{lcl}
X_5 &=&v_{1} \times (v_2 \cup v_3 \times (v_4 \cup v_5 \times v_6))\times (v_7 \cup (v_8 \cup v_9 \times v_{10})\times (v_{11} \cup
v_{12} \times v_{13}))\times \\
&&(v_{14}\cup (v_{15} \cup (v_{16} \cup v_{17} \times v_{18})\times (v_{19} \cup v_{20} \times v_{21}))\times  \\
&&(v_{22} \cup(v_{23} \cup v_{24} \times v_{25})\times v_{26}) )\times v_{27}
\end{array}
$$
Since by \cite{GKL20} it holds that $\chi_o(\g(X_5))=7$ the bound of Proposition \ref{prop-msp-7} is best possible.
\end{example}

Using the upper bound on the oriented chromatic number and  the recursive  structure of
msp-digraphs we achieve a linear time solution for
computing the oriented chromatic number of msp-digraphs.

\begin{theorem}[\cite{GKL20,GKL21}]\label{msp-ori-c}
Let $G$ be a msp-digraph. Then, the oriented
chromatic number of $G$ can be computed in linear time.
\end{theorem}

%%%%%%%%%%%%%%%%%%%%%%%%%%%%%%%%%%%%%%%%%%%%%%%%%%%%%%%%%%%%%%%%%%%%%%%%%%
\subsection{Oriented Arc-Colorings of Minimal Vertex Series-Parallel Digraphs}
%%%%%%%%%%%%%%%%%%%%%%%%%%%%%%%%%%%%%%%%%%%%%%%%%%%%%%%%%%%%%%%%%%%%%%%%%%

By Proposition \ref{prop-msp-7} and Observation \ref{obs-ind-chr}
we know the following bound on the oriented chromatic index of msp-digraphs.

\begin{corollary}\label{cor-msp-ind}
Let $G$ be an msp-digraph. Then, it holds that $\chi'_o(G)\leq 7$.
\end{corollary}

%\begin{problem}
%Is the bound of Corollary \ref{cor-msp-ind} tight?
%\end{problem}

For the optimality of the shown bound, we next give an example.

\begin{example}\label{ex-msp-6}
We recursively define msp-expressions $Y_i$ as follows. $Y_0$ defines a
single vertex graph and for $i\geq 1$
we define
$$Y_i= (Y_0 \cup Y_{i-1} \times Y_{i-1})$$
in order to define
$$X_6= Y_0 \times Y_0 \times Y_6 \times Y_0 \times Y_0.$$
$\gr(X_6)$ has 131 vertices and
satisfies $\chi'_o(\g(X_6))=7$, which was
found by a binary integer program (Remark \ref{rem-lp-oci}) using Gorubipy.
This implies that the bound of Corollary \ref{cor-msp-ind} is best possible.
\end{example}

%\begin{remark}\label{rem-lp-oci}
%
%> -Wie war nochmal der Graph mit 131 Knoten in Examle 4.10 bearbeitet
%> worden? War das mit einem Linearen Programm (ähnlich zu Remark 2.8)?
%> Wenn ja mit welchem Solver (das schreibt man oft dazu).
%
%Ja, mit einem Linearen Programm, Als Solver habe ich Gorubipy verwendet, also die Python-Schnittstelle des Gurobi %Optimizer. https://pypi.org/project/gurobipy/

%\begin{remark}\label{compute-oci}
%In order to compute the oriented chromatic index we
%have given a binary integer program in Remark \ref{rem-lp-oci}.
%From a theoretical point of view this implies the
%existence of an  $\fpt$-algorithm  for $\OCI$ w.r.t.\ parameter
%$n$, since integer linear programming is fixed-parameter tractable for the
%parameter ''number of variables''  \cite{Len83}.
%\end{remark}

%New results Lindemann ...
%
%x
%
%x
%
%x
%Let $G$ be some msp-digraph.
%To compute $\chi_o'(G)$ we apply Observation \ref{obs-nu-in} which allows us to 
%compute  $\chi_o(LD(G))$  which can be found by considering every possible color graph $H=(V_H,E_H)$ of $LD(G)$ %with $V_H \subseteq \{1, \dots, 7\}$. At first we do not care whether the color graphs are oriented or even contain loops. At the end we then check which is the color graph with the fewest vertices that is oriented. 

In order to compute the  oriented
chromatic index of an msp-digraph $G=(V,E)$ defined by an msp-expression $X$,
%Let $G$ be a minimal series parallel digraph given by some msp-expression $X$. 
we apply Observation \ref{obs-nu-in} which allows us to 
compute  the oriented chromatic number of $LD(G)$. This is done by computing 
triples $(H, \mathcal{L}, \mathcal{R})$, which are defined as follows.
Here $H = (V_H,E_H)$ is a color graph of $LD(G)$ with $V_H \subseteq \{1, \dots, 7\}$. At first we do not care whether the color graphs are oriented or even contain loops. At the end we then check which is the color graph with the fewest vertices that is oriented. 

The  color graph of $LD(G)$ for the parallel composition  $G = G_1 \cup G_2$ can easily be 
obtained from the color graphs $H_1=(V_1,E_1)$ of $LD(G_1)$ and  $H_2=(V_2,E_2)$ of $LD(G_2)$
by $H_1+H_2:=(V_1\cup V_2,E_1\cup E_2)$.

\begin{example}\label{ex-h-d}
In Figure \ref{lmspi_cup} the construction of a color graph of $LD(G)$ for $G = G_1 \cup G_2$ is illustrated.
\end{example}

\begin{figure}[ht]
\begin{center}
\includegraphics[width=0.9\textwidth]{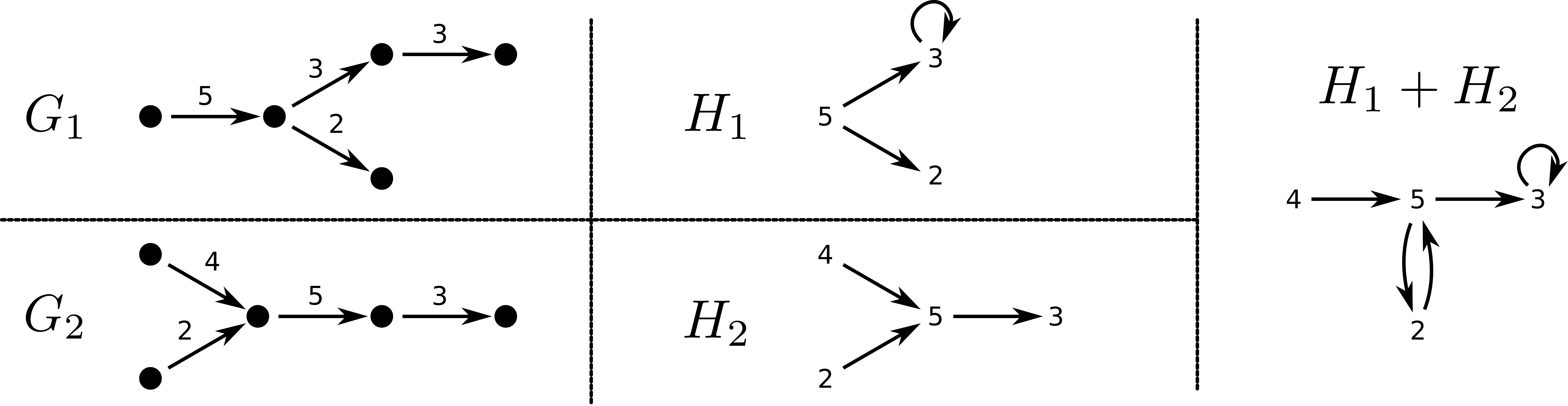}
\caption{$H_1$ is a color graph of $LD(G_1)$ and $H_2$ is a color graph of $LD(G_2)$, so $H_1 + H_2$ is a color graph of $LD(G_1 \cup G_2)$.}
\label{lmspi_cup}
\end{center}
\end{figure} 

The  color graph of $LD(G)$ for the series composition  $G = G_1 \times G_2$ obviously has new vertices 
which correspond to the new inserted edges by the series composition in $G$.

\begin{example}\label{ex-s1}
We consider  the series composition  $G = G_1 \times G_2$  in Figure \ref{lmspi_times}.
If $H_1$ is the color graph of the coloring of $LD(G_1)$  and $H_2$ is the color graph of the coloring of $LD(G_2)$, then we get several possible color graphs depending on how $x_1$, $x_2$, $y_1$ and $y_2$ are colored.

The color graphs are of the form $H_1 + H_2$ with additional vertices for the edges $(1,x_i)$, $(2,x_i)$, $(x_i,3)$, $(1,y_i)$, $(2,y_i)$ and $(y_i,4)$ with $i=1$ or $2$.

So in order to determine all possible color graphs of $G_1 \times G_2$, we have to know that $G_1$ has two sinks with leading edges colored with 1 and 2 and $G_2$ has one source with one leading edge colored with 3 and one source with one leading edge colored with 4.

As a first simplification, we can assume, without loss of generality, that $x_1 = x_2$ and $y_1 = y_2$. This applies because every oriented coloring with $x_1 \neq x_2$ also remains oriented when $x_1 = x_2$. So we do not need to know that $G_1$ has two sinks, which have incoming edges colored with 1 and 2, it is enough to know that $G_1$ has at least one such sink. 
\end{example}

\begin{figure}[ht]
\begin{center}
\includegraphics[width=0.3\textwidth]{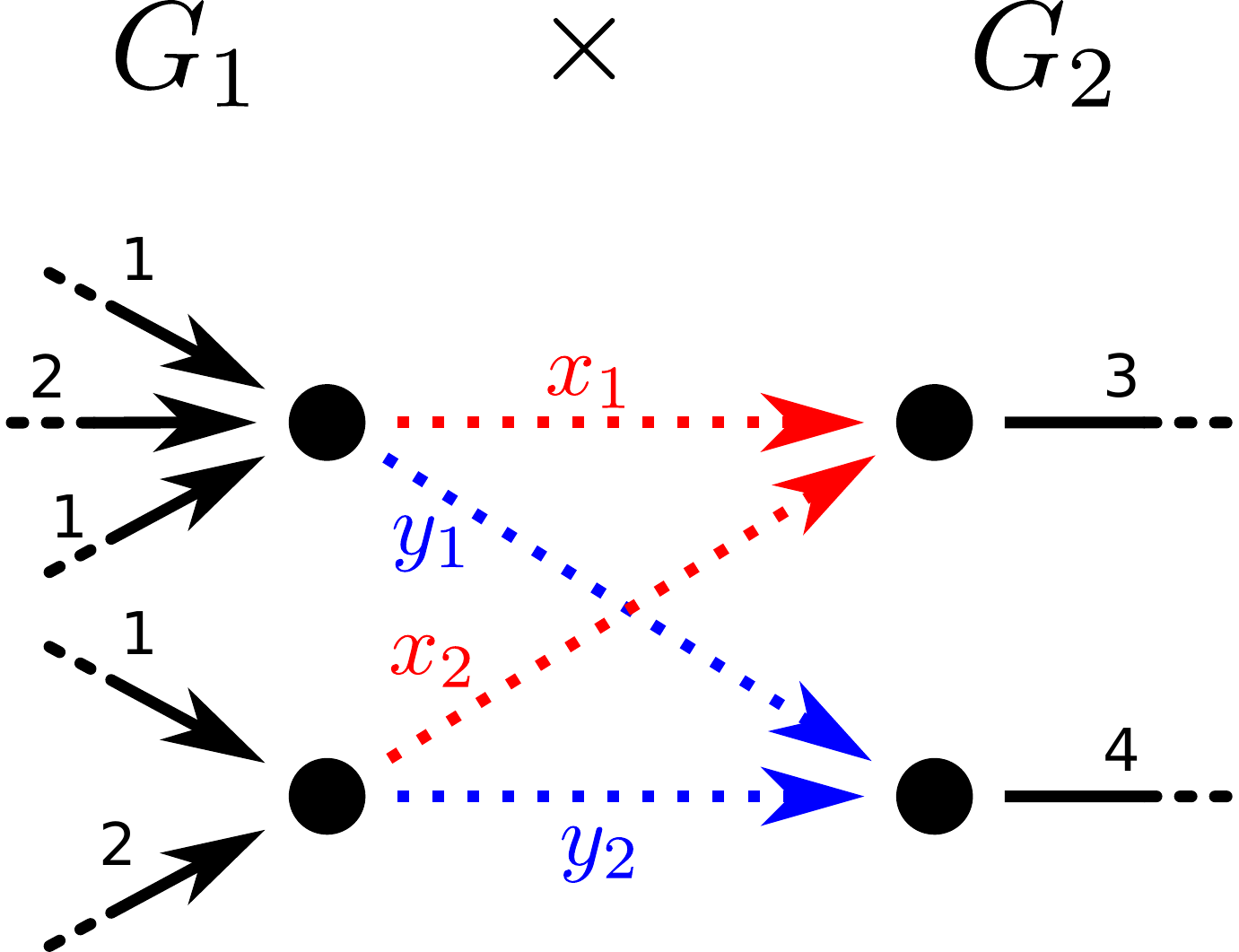}
\caption{Example of a series composition $G = G_1 \times G_2$ }
\label{lmspi_times}
\end{center}
\end{figure} 

The example shows that we need the colors of the incoming edges of sinks and the colors of the outgoing edges of sources in $G$.
In order to store these informations let $\mathcal{L}\subseteq \mathcal{P}(\{1,\dots,7\})$ 
be the set of the sets $L$ such that there is a source $\ell$ in $G$, such that $L$ is the set of outgoing edge colors in $\ell$ 
and $\mathcal{R}\subseteq \mathcal{P}(\{1,\dots,7\})$ be the set of the sets $R$ such that there is a sink $r$ in $G$, such that $R$ is the set of outgoing edge colors in $r$.

Assume we know the triples $((V_1,E_1),\mathcal{L}_1,\mathcal{R}_1)$ for $G_1$ and 
the triples  $((V_2,E_2),\mathcal{L}_2,\mathcal{R}_2)$ for $G_2$. For every
$R \in \mathcal{R}_1$ and every $L \in \mathcal{L}_2$ we define a new vertex  $u_{R,L}$ in $H$, which 
represents the color of all new edges that go from a sink $r$ of $G_1$, whose incoming edge colors are $R$, to a source $\ell$ of $G_2$, whose outgoing edge colors are $L$.

\begin{example}
In Figure \ref{lmspi_times} the color of the red edges would be $u_{\{1,2\},\{3\}}$ and the color of the blue edges $u_{\{1,2\}, \{4\}}$, representing two new vertices in the color graph $H$. 
\end{example}

For a color graph $H$ of $LD(G_1 \times G_2)$ we have the edges of $H_1$ and $H_2$ and additionally the edges $\bigcup_{R_i \in \mathcal{R}_1}(R_i \times \{u_{R_i,L_j} \mid L_j \in \mathcal{L}_2\})$ and $\bigcup_{L_j \in \mathcal{L}_2}(\{u_{R_i,L_j} \mid R_i \in \mathcal{R}_1\} \times L_j)$.

%So we need the sets: \[\mathcal{R}_1 = \left\lbrace R_1 \ \middle| \begin{array}{l}
%\text{there is a sink $r$ in $G_1$, so that $R_1$}\\
%\text{is the set of incoming edge colors in $r$}
%\end{array}\right\rbrace \subseteq \mathcal{P}(\{1,\dots,7\})\]
%
%\[\mathcal{L}_2 = \left\lbrace L_2 \ \middle| \begin{array}{l}
%\text{there is a source $\ell$ in $G_2$, so that $L_2$}\\
%\text{is the set of outgoing edge colors in $\ell$}
%\end{array}\right\rbrace \subseteq \mathcal{P}(\{1,\dots,7\})\]

\begin{example}
In Figure \ref{lmspi_times} the additional edges in color graph $H$ are $\{1,2\} \times \{u_{\{1,2\},\{3\}}, u_{\{1,2\},\{4\}}\}$ and $\{u_{\{1,2\},\{3\}}\} \times \{3\} \cup \{u_{\{1,2\},\{4\}}\} \times \{4\}$.
\end{example}

To obtain  the new set $\mathcal{L}$ for $\g(X_1 \times X_2)$ we have to distinguish between 
the following two cases. If $\emptyset \not \in \mathcal{L}_1$, i.e. $\g(X_1)$ has no isolated vertices, then $\mathcal{L} = \mathcal{L}_1$ applies. If $\emptyset \in \mathcal{L}_1$, then $\mathcal{L} = (\mathcal{L}_1 - \{\emptyset\}) \cup \{\{u_{\emptyset,L_j} \mid L_j \in \mathcal{L}_2\}\}$. The same can be concluded for the new 
set $\mathcal{R}$ for $\g(X_1 \times X_2)$. Thus, $\mathcal{R}$ and $\mathcal{L}$ can easily be constructed together with the color graph.

We store all these triples $(H, \mathcal{L}, \mathcal{R})$ in $F(X)$. 
In order to bound the size of $F(X)$ we recall that
the number of vertex labeled, i.e., the vertices are distinguishable
from each other, oriented graphs  on $n$
vertices is  $3^{\nicefrac{n(n-1)}{2}}$.
By  Corollary \ref{cor-msp-ind} we can conclude that
$$|F(X)|\leq 3^{\nicefrac{7(7-1)}{2}}\cdot 2^{2^7} \cdot 2^{2^7}\in \bigo(1)$$
which is independent of the size of $G$.

\begin{lemma}\label{le-e-msp}
\begin{enumerate}
\item For every $v\in V$ it holds that
\[F(v) = \{((\emptyset,\emptyset),\{\emptyset\},\{\emptyset\})\}.\]

\item Let $X = X_1 \cup X_2 $, then it holds that
$$
F(X) =  \left\lbrace (H_1 + H_2, \mathcal{L}_1 \cup \mathcal {L}_2, \mathcal {R}_1 \cup \mathcal {R}_2)  
\ \middle|    
\begin{array}{l}
(H_1, \mathcal {L}_1, \mathcal {R}_1) \in F (X_1), \\
(H_2, \mathcal {L}_2, \mathcal {R}_2) \in F (X_2) 
\end{array} 
\right\rbrace.
$$

\item Let $X = X_1 \times X_2 $, then it holds that
\[F(X) = \left\lbrace
(H , \mathcal{L}, \mathcal{R}) 
\ \middle| 
\begin{array}{l}
((V_1,E_1),\mathcal{L}_1,\mathcal{R}_1) \in F(X_1),\\ 
((V_2,E_2),\mathcal{L}_2,\mathcal{R}_2) \in F(X_2),\\
u_{R_i,L_j} \in \{1, \dots, 7\} \  \forall R_i \in \mathcal{R}_1, L_j \in \mathcal{L}_2, \text{ where}\\ 
U_{1,R_i} = \{u_{R_i,L_j} \mid L_j \in \mathcal{L}_2\},  \\ 
U_{2,L_j} = \{u_{R_i,L_j} \mid R_i \in \mathcal{R}_1\},  \\ 
U = \{u_{R_i,L_j} \mid R_i \in \mathcal{R}_1, L_j \in \mathcal{R}_2\},  \\ 
\mathcal{L} = \begin{cases}
\mathcal{L}_1 &\text{if } \emptyset \notin \mathcal{L}_1\\ 
(\mathcal{L}_1 - \{\emptyset\}) \cup \{U_{1,\emptyset}\} &\text{if } \emptyset \in \mathcal{L}_1
\end{cases},\\
\mathcal{R} = \begin{cases}
\mathcal{R}_2 &\text{if } \emptyset \notin \mathcal{R}_2\\ 
(\mathcal{R}_2 - \{\emptyset\}) \cup \{U_{2,\emptyset}\} &\text{if } \emptyset \in \mathcal{R}_2
\end{cases},\\
H = (V, E),\\
V = V_1 \cup V_2 \cup U,\\
E = E_1\cup E_2 \underset{R_i \in \mathcal{R}_1}{\bigcup} (R_i \times U_{1,R_i}) \underset{L_j \in \mathcal{L}_2}{\bigcup} (U_{2,L_j} \times L_j)
\end{array}  \right\rbrace.\]
\end{enumerate}
\end{lemma}

\begin{corollary}\label{cor1-ind}
There is an oriented edge $r$-coloring for an msp-digraph $G$
which is given by an msp-expression $X$ if and only if there is some $(H, \mathcal{L}, \mathcal{R}) \in F(X)$
such that color graph $H$ has $r$ vertices and is oriented.
Therefore, $\chi_o'(G) = \min \{|V_H| \mid ((V_H,E_H), \mathcal{L}, \mathcal{R}) \in F(X) \text{ and } (V_H,E_H) \text{ is oriented}\}$.
\end{corollary}

\begin{theorem}\label{msp-ori-a-c}
Let $G$ be a msp-digraph. Then, the oriented
chromatic index of $G$ can be computed in linear time.
\end{theorem}

\begin{proof}
Let $G=(V,E)$ be an msp-digraph with $n=|V|$ vertices and $m=|E|$ edges and let  $T$ be an msp-tree for $G$
with root $r$. For a vertex $u$ of $T$ we denote by $T_u$
the subtree rooted at $u$ and by $X_u$ the esp-expression defined by $T_u$.

For computing the oriented chromatic index for an msp-digraph $G$, we traverse msp-tree $T$ in bottom-up order.
For every vertex $u$ of $T$ we can compute $F(X_u)$ by
following the rules given in Lemma \ref{le-e-msp}. By Corollary \ref{cor1-ind} we can solve our
problem using $F(X_r)=F(X)$.

An msp-tree $T$ can be computed in $\bigo(n+m)$ time from $G$, see \cite{VTL82}.
By Lemma \ref{le-e-msp} we obtain the following running times.
\begin{itemize}
\item
For every vertex $v\in V$ set  $F(v)$ is computable in $\bigo(1)$ time.

\item
For every two msp-expressions $X_1$ and $X_2$  set
$F(X_1  \cup X_2)$  and   set
$F(X_1  \times X_2)$  can be computed in  $\bigo(1)$ time from
$F(X_1)$ and $F(X_2)$.
This is true since there are only a constant number of color graphs with nodes from $\{1, \ldots, 7\}$, just as there are only a limited number of different $\mathcal{L}$ and $\mathcal{R}$ and so the input size  is constant.

%
%\item
%For every two msp-expressions $X_1$ and $X_2$ can be computed in   $\bigo(1)$ time from $F(X_1)$ and $F(X_2)$.
\end{itemize}

Since $T$ consists of $n$ leaves and $n-1$ inner vertices,
the overall running time is in $\bigo(n+m)$. 
\end{proof}

%%%%%%%%%%%%%%%%%%%%%%%%%%%%%%%%%%%%%%%%%%%%%%%%%%%%%%%%%%%%%%%%%%%%%%%%%%
%%%%%%%%%%%%%%%%%%%%%%%%%%%%%%%%%%%%%%%%%%%%%%%%%%%%%%%%%%%%%%%%%%%%%%%%%%
\section{Conclusions and outlook}
%%%%%%%%%%%%%%%%%%%%%%%%%%%%%%%%%%%%%%%%%%%%%%%%%%%%%%%%%%%%%%%%%%%%%%%%%%
%%%%%%%%%%%%%%%%%%%%%%%%%%%%%%%%%%%%%%%%%%%%%%%%%%%%%%%%%%%%%%%%%%%%%%%%%%

In this paper we showed
tight upper bounds for the oriented chromatic number
and the oriented chromatic index of edge series-parallel digraphs and 
minimal series-parallel digraphs.
Furthermore, we introduced  linear time solutions for 
computing the  oriented chromatic number of edge series-parallel digraphs and the  oriented chromatic index of minimal series-parallel digraphs.

%alt:
%
%The presented methods allow us to re-prove the known bound of $7$ for the oriented chromatic number
%and the oriented chromatic index of series-parallel digraphs and we showed
%that these bounds are tight even for series-parallel digraphs.
%Furthermore, we give linear time solutions for computing the
%oriented chromatic number and the oriented chromatic index of
%series-parallel digraphs.

The existence of graph classes of arbitrary large vertex degree but bounded
oriented chromatic index, such as msp-digraph and esp-digraphs,
implies that Vizings Theorem \cite{Viz64} can not
be carried over to the oriented chromatic index.

In our future work we want to analyze
the existence of polynomial time
algorithms for computing  
the oriented chromatic index and  oriented chromatic
number of orientations of  series-parallel graphs efficiently
which would lead to generalizations of
Theorem \ref{esp-index-com} and Theorem \ref{esp-num-com}.

Furthermore, it remains to extend the results for $\OCN$ and $\OCI$ to graphs of bounded
directed clique-width. As a starting point we considered the parameterized complexity of $\OCN_r$ parameterized
by directed clique-width  in \cite{GKL21a}.

%%%%%%%%%%%%%%%%%%%%%%%%%%%%%%%%%%%%%%%%%%%%%%%%%%%%%%%%%%%%%%%%%%%%%%%%%%
%%%%%%%%%%%%%%%%%%%%%%%%%%%%%%%%%%%%%%%%%%%%%%%%%%%%%%%%%%%%%%%%%%%%%%%%%%

%%%%%%%%%%%%%%%%%%%%%%%%%%%%%%%%%%%%%%%%%%%%%%%%%%%%%%%%%%%%%%%%%%%%%%
\section*{Acknowledgements} \label{sec-a}
%%%%%%%%%%%%%%%%%%%%%%%%%%%%%%%%%%%%%%%%%%%%%%%%%%%%%%%%%%%%%%%%%%%%%%

The work of the second and third author was supported
by the Deutsche
Forschungsgemeinschaft (DFG, German Research Foundation) -- 388221852

%%%%%%%%%%%%%%%%%%%%%%%%%%%%%%%%%%%%%%%%%%%%%%%%%%%%%%%%%%%%%%%%%%%%%%
%%%%%%%%%%%%%%%%%%%%%%%%%%%%%%%%%%%%%%%%%%%%%%%%%%%%%%%%%%%%%%%%%%%%%%

%%%%%%%%%%%%%%%%%%%%%%%%%%%%%%%%%%%%%%%%%%%%%%%%%%%%%%%%%%%%%%%%%%%%%%%%%%
%%%%%%%%%%%%%%%%%%%%%%%%%%%%%%%%%%%%%%%%%%%%%%%%%%%%%%%%%%%%%%%%%%%%%%%%%%

%\bibliographystyle{alpha}
%\bibliography{/home/gurski/bib.bib}

\begin{thebibliography}{GKL21b}

\bibitem[BJG09]{BG09}
J.~Bang-Jensen and G.~Gutin.
\newblock {\em Digraphs. {T}heory, {A}lgorithms and {A}pplications}.
\newblock Springer-Verlag, Berlin, 2009.

\bibitem[BJG18]{BG18}
J.~Bang-Jensen and G.~Gutin, editors.
\newblock {\em Classes of Directed Graphs}.
\newblock Springer-Verlag, Berlin, 2018.

\bibitem[BLS99]{BLS99}
A.~Brandst\"adt, V.B. Le, and J.P. Spinrad.
\newblock {\em Graph Classes: A Survey}.
\newblock SIAM Monographs on Discrete Mathematics and Applications. SIAM,
  Philadelphia, 1999.

\bibitem[Bod98]{Bod98}
H.L. Bodlaender.
\newblock A partial $k$-arboretum of graphs with bounded treewidth.
\newblock {\em Theoretical Computer Science}, 209:1--45, 1998.

\bibitem[CD06]{CD06}
J.-F. Culus and M.~Demange.
\newblock Oriented coloring: Complexity and approximation.
\newblock In {\em Proceedings of the Conference on Current Trends in Theory and
  Practice of Computer Science (SOFSEM)}, volume 3831 of {\em LNCS}, pages
  226--236. Springer-Verlag, 2006.

\bibitem[CLSB81]{CLS81}
D.G. Corneil, H.~Lerchs, and L.~Stewart-Burlingham.
\newblock Complement reducible graphs.
\newblock {\em Discrete Applied Mathematics}, 3:163--174, 1981.

\bibitem[Cou94]{Cou94}
B.~Courcelle.
\newblock The monadic second-order logic of graphs {VI}: On several
  representations of graphs by relational structures.
\newblock {\em Discrete Applied Mathematics}, 54:117--149, 1994.

\bibitem[Dai80]{Dai80}
D.P. Dailey.
\newblock Uniqueness of colorability and colorability of planar 4-regular
  graphs are {NP}-complete.
\newblock {\em Discrete Mathematics}, 30(3):289--293, 1980.

\bibitem[DS14]{DS14}
J.~Dybizba{\'n}ski and A.~Szepietowski.
\newblock The oriented chromatic number of {H}alin graphs.
\newblock {\em Information Processing Letters}, 114(1-2):45--49, 2014.

\bibitem[Epp92]{Epp92}
D.~Eppstein.
\newblock Parallel recognition of series-parallel graphs.
\newblock {\em Information and Computation}, 98(1):41--55, 1992.

\bibitem[GKL20]{GKL20}
F.~Gurski, D.~Komander, and M.~Lindemann.
\newblock Oriented coloring of msp-digraphs and oriented co-graphs.
\newblock In {\em Proceedings of the International Conference on Combinatorial
  Optimization and Applications (COCOA)}, volume 12577 of {\em LNCS}, pages
  743--758. Springer-Verlag, 2020.

\bibitem[GKL21a]{GKL21a}
F.~Gurski, D.~Komander, and M.~Lindemann.
\newblock Efficient computation of the oriented chromatic number of recursively
  defined digraphs.
\newblock {\em Theoretical Computer Science}, 890:16--35, 2021.

\bibitem[GKL21b]{GKL21}
F.~Gurski, D.~Komander, and M.~Lindemann.
\newblock Homomorphisms to digraphs with large girth and oriented colorings of
  minimal series-parallel digraphs.
\newblock In {\em Proceedings of the International Workshop on Algorithms and
  Computation (WALCOM)}, volume 12635 of {\em LNCS}, pages 182--194.
  Springer-Verlag, 2021.

\bibitem[GKR19]{GKR19d}
F.~Gurski, D.~Komander, and C.~Rehs.
\newblock Oriented coloring on recursively defined digraphs.
\newblock {\em Algorithms}, 12(4):87, 2019.

\bibitem[GO15]{GO15}
G.~Guegan and P.~Ochem.
\newblock Complexity dichotomy for oriented homomorphism of planar graphs with
  large girth.
\newblock {\em Theoretical Computer Science}, 596:142--148, 2015.

\bibitem[Gur14]{Gur14}
F.~Gurski.
\newblock Efficient binary linear programming formulations for boolean
  functions.
\newblock {\em Statistics, Optimization and Information Computing},
  2(4):274--279, 2014.

\bibitem[HN60]{HN60}
F.~Harary and R.Z. Norman.
\newblock Some properties of line digraphs.
\newblock {\em Rend. Circ. Mat. Palermo}, 9(2):161--168, 1960.

\bibitem[HY87]{HY87}
X.~He and Y.~Yesha.
\newblock Parallel recognition and decomposition of two terminal series
  parallel graphs.
\newblock {\em Information and Computation}, 75:15--38, 1987.

\bibitem[KM04]{KMG04}
W.F. Klostermeyer and G.~MacGillivray.
\newblock Homomorphisms and oriented colorings of equivalence classes of
  oriented graphs.
\newblock {\em Discrete Mathematics}, 274:161--172, 2004.

\bibitem[LGK21]{LGK21}
M.~Lindemann, F.~Gurski, and D.~Komander.
\newblock Oriented vertex and arc coloring of edge series-parallel digraphs.
\newblock In {\em Operations Research Proceedings (OR 2021), Selected Papers}.
  Springer-Verlag, 2021.
\newblock to appear.

\bibitem[Mar13]{Mar13}
T.H. Marshall.
\newblock Homomorphism bounds for oriented planar graphs of given minimum
  girth.
\newblock {\em Graphs and Combin.}, 29:1489--1499, 2013.

\bibitem[Mar15]{Mar15}
T.H. Marshall.
\newblock On oriented graphs with certain extension properties.
\newblock {\em Ars Combinatoria}, 120:223--236, 2015.

\bibitem[OP14]{OP14}
P.~Ochem and A.~Pinlou.
\newblock Oriented coloring of triangle-free planar graphs and 2-outerplanar
  graphs.
\newblock {\em Graphs and Combin.}, 30:439--453, 2014.

\bibitem[OPS08]{OPS08}
P.~Ochem, A.~Pinlou, and E.~Sopena.
\newblock On the oriented chromatic index of oriented graphs.
\newblock {\em Journal of Graph Theory}, 57(4):313--332, 2008.

\bibitem[PS06]{PS06}
A.~Pinlou and E.~Sopena.
\newblock Oriented vertex and arc coloring of outerplanar graphs.
\newblock {\em Information Processing Letters}, 100:97--104, 2006.

\bibitem[Sey90]{Sey90}
P.D. Seymour.
\newblock Colouring series-parallel graphs.
\newblock {\em Combinatorica}, 10(4):379--392, 1990.

\bibitem[Sop97]{Sop97}
E.~Sopena.
\newblock The chromatic number of oriented graphs.
\newblock {\em Journal of Graph Theory}, 25:191--205, 1997.

\bibitem[Sop16]{Sop16}
E.~Sopena.
\newblock Homomorphisms and colourings of oriented graphs: {A}n updated survey.
\newblock {\em Discrete Mathematics}, 339:1993--2005, 2016.

\bibitem[Val78]{Val78}
J.~Valdes.
\newblock Parsing flowcharts and series-parallel graphs.
\newblock Technical Report STAN-CS-78-682, Computer Science Department,
  Stanford University, Stanford, California, 1978.

\bibitem[Viz64]{Viz64}
V.G. Vizing.
\newblock On an estimate of the chromatic class of a p-graph.
\newblock {\em Metody Diskret. Analiz.}, 3:9--17, 1964.

\bibitem[VTL82]{VTL82}
J.~Valdes, R.E. Tarjan, and E.L. Lawler.
\newblock The recognition of series-parallel digraphs.
\newblock {\em SIAM Journal on Computing}, 11:298--313, 1982.

\end{thebibliography}

\end{document}